\documentclass[journal]{IEEEtran}
\usepackage{amsmath,amsfonts}
\usepackage{array}
\usepackage[caption=false,font=normalsize,labelfont=sf,textfont=sf]{subfig}
\usepackage{textcomp}
\usepackage{stfloats}
\usepackage{url}
\usepackage{verbatim}
\usepackage{graphicx}
\usepackage{cite}
\usepackage{xcolor}
\usepackage{tabularx} 
\usepackage[ruled, lined, linesnumbered, commentsnumbered, longend]{algorithm2e}

\usepackage{bm}
\newtheorem{theorem}{Theorem}

\newtheorem{lemma}{Lemma}

\DeclareMathOperator*{\argmin}{arg\,min}

\makeatletter
\let\save@mathaccent\mathaccent
\newcommand*\if@single[3]{%
  \setbox0\hbox{${\mathaccent"0362{#1}}^H$}%
  \setbox2\hbox{${\mathaccent"0362{\kern0pt#1}}^H$}%
  \ifdim\ht0=\ht2 #3\else #2\fi
  }
\newcommand*\rel@kern[1]{\kern#1\dimexpr\macc@kerna}
\newcommand*\widebar[1]{\@ifnextchar^{{\wide@bar{#1}{0}}}{\wide@bar{#1}{1}}}
\newcommand*\wide@bar[2]{\if@single{#1}{\wide@bar@{#1}{#2}{1}}{\wide@bar@{#1}{#2}{2}}}
\newcommand*\wide@bar@[3]{%
  \begingroup
  \def\mathaccent##1##2{%
    \let\mathaccent\save@mathaccent
    \if#32 \let\macc@nucleus\first@char \fi
    \setbox\z@\hbox{$\macc@style{\macc@nucleus}_{}$}%
    \setbox\tw@\hbox{$\macc@style{\macc@nucleus}{}_{}$}%
    \dimen@\wd\tw@
    \advance\dimen@-\wd\z@
    \divide\dimen@ 3
    \@tempdima\wd\tw@
    \advance\@tempdima-\scriptspace
    \divide\@tempdima 10
    \advance\dimen@-\@tempdima
    \ifdim\dimen@>\z@ \dimen@0pt\fi
    \rel@kern{0.6}\kern-\dimen@
    \if#31
      \overline{\rel@kern{-0.6}\kern\dimen@\macc@nucleus\rel@kern{0.4}\kern\dimen@}%
      \advance\dimen@0.4\dimexpr\macc@kerna
      \let\final@kern#2%
      \ifdim\dimen@<\z@ \let\final@kern1\fi
      \if\final@kern1 \kern-\dimen@\fi
    \else
      \overline{\rel@kern{-0.6}\kern\dimen@#1}%
    \fi
  }%
  \macc@depth\@ne
  \let\math@bgroup\@empty \let\math@egroup\macc@set@skewchar
  \mathsurround\z@ \frozen@everymath{\mathgroup\macc@group\relax}%
  \macc@set@skewchar\relax
  \let\mathaccentV\macc@nested@a
  \if#31
    \macc@nested@a\relax111{#1}%
  \else
    \def\gobble@till@marker##1\endmarker{}%
    \futurelet\first@char\gobble@till@marker#1\endmarker
    \ifcat\noexpand\first@char A\else
      \def\first@char{}%
    \fi
    \macc@nested@a\relax111{\first@char}%
  \fi
  \endgroup
}
\makeatother

\begin{document}

\title{\Huge Energy-Efficient Online Scheduling for Wireless Powered Mobile Edge Computing Networks}

\author{
  Xingqiu He,
  Chaoqun You,
  Yuzhi Yang,
  Zihan Chen, 
  Yuhang Shen, \\
  Tony Q. S. Quek,~\IEEEmembership{Fellow,~IEEE},
  Yue Gao,~\IEEEmembership{Fellow,~IEEE}
\thanks{X. He, C. You, and Y. Gao are with the Institute of Space Internet, Fudan University, Shanghai,
China (emails: hexqiu@gmail.com, chaoqunyou@gmail.com, gao.yue@fudan.edu.cn).}
\thanks{Y. Yang is with the 6G Research Center, Khalifa University, 127788 Abu Dhabi, UAE (e-mail: yuzhi.yang@ku.ac.ae).}
\thanks{Z. Chen is with the Singapore University of Technology and Design, Singapore 487372 (e-mail: zihan\_chen@mymail.sutd.edu.sg).}
\thanks{Y. Shen is with the Department of Computer and Information Engineering, Khalifa University, Abu Dhabi, UAE (e-mail: 100066921@ku.ac.ae).}
\thanks{T. Q. S. Quek is with the Singapore University of Technology and Design,
Singapore 487372, and also with the Yonsei Frontier Lab, Yonsei University,
South Korea (e-mail: tonyquek@sutd.edu.sg).}
}



\maketitle

\begin{abstract}
Wireless Powered Mobile Edge Computing (WPMEC) integrates mobile edge computing (MEC) with wireless power transfer (WPT) 
to simultaneously extend the operational lifetime and enhance the computational capability of wireless devices (WDs). 
In WPMEC systems, WPT and computation offloading compete for limited wireless resources, which makes their joint scheduling particularly challenging.
In this paper, we investigate the energy-efficient online scheduling problem for WPMEC networks with multiple WDs and multiple access points (APs). 
Based on Lyapunov optimization, we develop an online optimization framework 
that transforms the original stochastic problem into deterministic per-slot optimization problems. 
To reduce computational complexity, we introduce the concept of marginal energy efficiency and derive an associated optimality condition, 
based on which a relax–then–adjust approach is proposed to efficiently obtain feasible solutions. 
For the resulting non-convex computation offloading subproblem, we analyze the structural properties of its 
optimal solution and transform it into an assignment problem that can be solved efficiently. 
We further provide theoretical performance guarantees for both the per-slot and long-term solution,
establishing a fundamental trade-off between latency and energy consumption. 
To improve practical performance, additional mechanisms are introduced to balance the magnitudes of different queues and reduce latency without increasing energy consumption. 
Extensive simulation results demonstrate the effectiveness and robustness of the proposed algorithm under various system settings.
\end{abstract}


\section{Introduction}
With the rapid development in recent years, the Internet of Things (IoT) technology has played an important role
in the intelligent and autonomous control of many industrial and commercial systems, such as smart grid and smart cities \cite{al2015internet}.
Due to the stringent size constraint and production cost consideration,
the ubiquitously deployed IoT Wireless Devices (WDs) usually have restricted computation capability and finite battery capacity,
which severely degrades the Quality of Service (QoS) experienced by users.
To handle the two fundamental performance limitations, Wireless Powered Mobile Edge Computing (WPMEC) has been proposed as a novel paradigm
that combines the advantages of Wireless Power Transfer (WPT) and Mobile Edge Computing (MEC) \cite{wang2023wireless}.

As a promising approach that provides sustainable energy supply for WDs, WPT utilizes dedicated energy transmitters to broadcast radio frequency (RF) signals \cite{zhang2018wireless}.
The received RF signals can be converted to electricity by energy harvesting circuits and used to charge WDs continuously.
On the other hand, MEC is a newly emerged computing paradigm that enables WDs to offload their computation tasks to nearby edge servers located at the edge of radio access networks \cite{mao2017survey}.
As the integration of both techniques, WPMEC charges WDs with WPT and alleviates WDs' computation workloads with MEC.
As a result, the WDs' device lifetime and computation capacity are simultaneously improved, which leads to significantly enhanced user experiences.

While WPMEC provides significant improvement, it results in a much more complicated scheduling problem.
First, WDs are merely powered by WPT, hence the energy-consuming operations of WDs, such as local computing and computation offloading,
are strictly restricted by the amount of harvested energy.
Second, due to the half-duplex of access points (APs) or WDs,
the WPT process and the computation offloading process have to be separated in time horizon.
Therefore, the WPT and computation offloading process are deeply interwined and requires joint optimization.

In this paper, we investigate the scheduling of WPT and computation offloading in WPMEC networks.
Different from most existing works that focus on single AP or a single time slot, 
we consider a general multi-AP-multi-WD scenario and formulate our problem under an online setting,
where the surplus energy from previous time slots can be stored in the battery and utilized later.
This brings new challenges in algorithm design as 
(i) the wireless resource allocation among WDs and APs are substantially more complicated and 
(ii) decision variables in different time slots are deeply coupled through both energy and data queue dynamics.

Our objective is to minimize the long-term energy consumption and latency by jointly optimizing the resource allocation in WPT 
and computation offloading stages.
To guarantee the energy-harvesting efficiency, only one AP is allowed to broadcast RF energy at the same time.
By choosing different APs for WPT in turn, the proposed approach can alleviate the ``double near-far'' effect \cite{ji2018energy}---which 
occurs because a farther device harvests less energy from the AP but spends more power to communicate in longer distances---because 
each WD has a chance to harvest energy from a closer AP.
Our main contributions are summarized as follows.
\begin{itemize}
    \item We formulate the energy-efficient online scheduling problem for WPMEC networks with multiple APs and WDs.
      To reduce the computational complexity of the proposed algorithm, we introduce the concept of \emph{marginal energy efficiency} 
      and establish an associated optimality condition. 
      Leveraging this condition, we propose a relax-then-adjust framework, in which the original problem is first relaxed 
      and decomposed into several independent subproblems. 
      After solving each subproblem, the resulting decision variables are adjusted according to the optimality condition to construct a feasible solution.
    \item To solve the non-convex computation offloading subproblem, we analyze the structure of its optimal solution and 
      transform it to an assignment problem, which can be efficiently solved using the Hungarian algorithm. 
      We further prove that the optimality gap between the proposed solution and the global optimum of the single-slot problem is upper bounded by a constant. 
      Building on this result, we establish an $\mathcal{O}(V)$--$\mathcal{O}(1/V)$ trade-off between latency and energy consumption in the online setting.
    \item To further improve practical performance, we introduce a virtual battery capacity together with scalar weighting coefficients to 
      balance the large disparity between the magnitudes of energy queues and data queues. 
      In addition, we incorporate place-holder backlogs to reduce latency without increasing energy consumption. 
      An estimation-based mechanism is proposed to automatically adapt the length of the place-holder backlogs, 
      eliminating the need for prior knowledge of the optimal Lagrange multipliers.
    \item Extensive simulations are conducted to validate the effectiveness of the proposed algorithm. 
      Numerical results demonstrate that the long-term energy consumption and latency are significantly reduced under a wide range of system settings. 
      All source codes are made publicly available\footnote{https://github.com/XingqiuHe/wpmec} to facilitate reproducibility and enable further research based on this work.
\end{itemize}

The rest of the paper is organized as follows.
In Section \ref{section:related work}, we review related works.
In Section \ref{section:system model}, we introduce the system model and formulate the online scheduling problem.
An online algorithm that jointly optimizes the WPT and computation scheduling is proposed in Section \ref{section:algorithm design}.
Two effective approaches are proposed in Section \ref{section:improvement} to further improve the performance of our algorithm.
The simulations and related numerical results are presented in Section \ref{section:simulation} and we conclude our paper in Section \ref{section:conclusion}.

\section{Related Work} \label{section:related work}
Recent advances in WPT technology \cite{krikidis2014simultaneous, zhang2019wiress, clerckx2021wireless} have enabled the design of WPMEC networks. 
Early studies such as \cite{mao2016dynamic} and \cite{you2016energy} investigated scheduling for a single WD, 
jointly considering energy-harvesting and computation-offloading operations. 
Subsequent research extended these models to multi-WD settings. 
For example, \cite{hu2018wireless, ji2018energy} examined WPMEC networks with near–far user asymmetries and 
addressed the ``double near–far'' effect. 
Other works focused on system-level performance: 
\cite{bi2018computation, zhou2018computation} sought to maximize the aggregate computation rate across WDs, 
with \cite{zhou2018computation} employing an unmanned aerial vehicle (UAV) as a mobile energy transmitter. 
More recently, \cite{zheng2024minimization} considered minimizing task completion time in WPMEC networks where tasks can be further offloaded to a cloud server.

Most of the above studies adopt binary offloading, wherein each task is indivisible and is either executed locally on the WD or fully offloaded to the AP. 
Beyond binary offloading, partial offloading has also been explored, allowing tasks to be partitioned and processed concurrently at the WD and the AP. 
Specifically, \cite{wang2017joint} proposed a unified multi-WD MEC–WPT framework that jointly optimizes energy beamforming, 
computation, and offloading to minimize AP energy consumption under latency constraints. 
In \cite{zeng2023joint}, a cooperative scheme based on non-orthogonal multiple access (NOMA) was developed to mitigate the double 
near–far effect via joint communication–computation cooperation.
The authors in \cite{zhou2020computation, chen2024computation, wu2022residual} considered both binary and partial offloading,
focusing on computation efficiency maximization under the max-min fairness criterion \cite{zhou2020computation}, 
minimizing computation completion time under dynamic channel conditions \cite{chen2024computation},
and maximizing redidual energy under a non-linear energy harvesting model \cite{wu2022residual}.

A common limitation of these works is their ``one-shot'' optimization perspective, whereby scheduling is optimized for a single time 
interval without accounting for temporal coupling. 
In practice, WPMEC systems operate sustainably, and control actions across time are interdependent.
For example, the energy harvested previously can be stored in the battery for future use. 
Motivated by this, many recent works try to solve the scheduling problem in online settings.
For a single WD and single AP, \cite{wang2020optimal} minimized energy consumption by jointly optimizing the WPT power and offloading decisions.
The model was extended in \cite{wu2019online} to include multiple WDs with a long-term throughput objective. 
Using Lyapunov optimization, \cite{mao2019energy} derived an online algorithm that characterizes a theoretical energy–delay trade-off. 
To support real-time control under fast fading environments, \cite{huang2019deep} proposed a deep reinforcement learning approach that yields 
near-optimal solutions in large-scale WPMEC networks. 
In addition, \cite{wang2020real} developed an offline-inspired online resource-allocation framework that jointly optimizes 
energy beamforming, task execution, and offloading under causal task and channel information to reduce system energy consumption.

It is worth noting that the above works predominantly consider single-AP scenarios. 
To the best of our knowledge, only \cite{zhu2020computation} and \cite{deng2022wireless} have investigated WPMEC networks with multiple APs. 
Specifically, in \cite{zhu2020computation}, an approximation algorithm was developed to maximize the proportion of tasks completed within their deadlines.
However, the study is limited to a single time slot and assumes fixed WD transmission power. 
The work in \cite{deng2022wireless} designed a Lyapunov-based dynamic control algorithm for WPMEC systems with both task and energy queue dynamics, 
achieving near-optimal throughput with reduced queue backlog. 
Nevertheless, \cite{deng2022wireless} makes several simplifying assumptions: it allows simultaneous WPT and offloading processes and presumes 
that WDs lack local computation capabilities, requiring all tasks to be offloaded to APs.
In contrast, this paper studies online scheduling in multi-AP, multi-WD WPMEC networks where tasks can 
be flexibly executed either locally or at APs, thereby achieving a more realistic and general system model.

\section{System Model and Problem Formulation} \label{section:system model}
In this section, we present the considered system model and 
formulate the energy-efficient online scheduling problem, where the control decisions
for WPT, local computing, and computation offloading are jointly optimized.
The major notations used in this paper are summarized in Table \ref{tab:notation}.

\begin{table}[!t]
\renewcommand{\arraystretch}{1.2}
\caption{Major Notations}
\label{tab:notation}
\centering
\begin{tabularx}{0.99\linewidth}{l l}
\hline
\textbf{Notation} & \textbf{Description}\\
\hline
$N, \mathcal{N}$ & Number of WDs and corresponding index set. \\
$M, \mathcal{M}$ & Number of APs and corresponding index set. \\
$T, \mathcal{T}$ & Number of time slots and corresponding index set. \\
$\Delta t$ & Duration of each time slot. \\
$h^D_{ij}(t), h^U_{ij}(t)$ & Downlink and uplink channel gain between WD $i$ \\
                           & and AP $j$. \\
$a^T_j(t)$ & Whether AP $j$ broadcast RF energy in slot $t$. \\
$a^O_{ij}(t)$ & Whether WD $i$ is communicating with AP $j$ in slot $t$. \\
$P^T_j(t), \tau^T_j(t)$ & Transmission power and time of AP $j$ for WPT. \\
$P^O_i(t), \tau^O_i(t)$ & Transmission power and offloading time of WD $i$ \\
                        & for computation offloading. \\
$E^T_j(t)$ & Energy consumption of AP $j$ due to WPT. \\
$E^C_j(t)$ & Energy consumption for the computation at AP $j$. \\
$E^O_i(t)$ & Energy consumption for offloading of WD $i$. \\
$E^H_i(t)$ & Energy harvested by WD $i$ during time slot $t$. \\
$E^L_i(t), E^O_i(t)$ & Energy consumption for local computing and \\
                     & computation offloading. \\
$D^L_i(t), D^O_i(t)$ & Amount of locally processed and offloaded data. \\
$B^{max}_i, B_i(t)$ & Maximum and remaining power of WD $i$. \\
$A_i(t)$ & Arrived computation data of WD $i$ in slot $t$. \\
$Q_i(t)$ & Length of queueing data at WD $i$. \\
$f_i(t)$ & CPU frequency of WD $i$ in slot $t$. \\
\hline
\end{tabularx}
\end{table}

\subsection{Network Model}
As shown in Fig. \ref{fig:wpmec}, we consider a WPMEC network consisting of
$N$ WDs and $M$ APs, where each AP is integrated with an RF energy transmitter and a MEC server.
The set of WDs and APs are denoted by $\mathcal{N}$ and $\mathcal{M}$, respectively.
APs are assumed to have a stable power supply and broadcast RF to charge WDs.
The energy harvested by each WD is stored in a rechargeable battery, which is used to power its computational and communication operations.
In this work, we assume APs are operated at orthogonal frequency bands,
hence the WPT and wireless communications of \emph{different} APs can be operated simultaneously without mutual interference,
but the WPT and wireless communications (for offloading) of the \emph{same} AP
cannot be performed simultaneously and the TDMA protocol is applied to avoid mutual interference.

The time horizon is divided into $T$ slots with equal length $\Delta t$,
where each slot consists of four phases, i.e., WPT, computation offloading, edge computing, and result downloading,
as illustrated in Fig. \ref{fig:wpmec}.
The set of time slots is dentoed by $\mathcal{T}$.
Note that we have assumed WDs have simultaneous wireless information and power transfer (SWIPT) abilities \cite{perera2017simultaneous, clerckx2021wireless},
e.g., WDs have separate receivers for information transmission and energy harvesting.
Therefore, $\mathit{WD}_1$ and $\mathit{WD}_4$ can communicate with $\mathit{AP}_1$ and $\mathit{AP}_3$ while harvesting energy from $\mathit{AP}_2$.
Our model and algorithm can be easily extended to accommodate WDs without SWIPT capabilities,
as will be discussed in Section \ref{subsection:problem_formulation}.

Given that MEC servers possess significantly higher computational capabilities than WDs, 
and that the results of edge computation typically have small data sizes, the time required for task execution at the edge and result downloading is negligible
\cite{zhu2020computation, wang2017joint, zhou2020computation, wu2019online, mao2019energy}.
Consequently, this study focuses exclusively on the durations of the first two operational phases in the system model.
Let $h^{D}_{ij}(t)$ and $h^{U}_{ij}(t)$ denote the downlink and uplink channel gains between WD $i$ and AP $j$ during time slot $t$. 
If WD $i$ lies outside the communication range of AP $j$, the corresponding channel gain is assumed to be zero. 
Following the assumptions in \cite{wang2020optimal, wu2019online, mao2019energy}, 
all wireless channels are modeled as quasi-static flat-fading, i.e., the channel state remains constant within each time slot but may vary independently across different slots.
Our objective is to minimize the total energy consumption for processing the computational tasks of all WDs.

\begin{figure}[!t]
    \centering
    \includegraphics[width=0.45\textwidth]{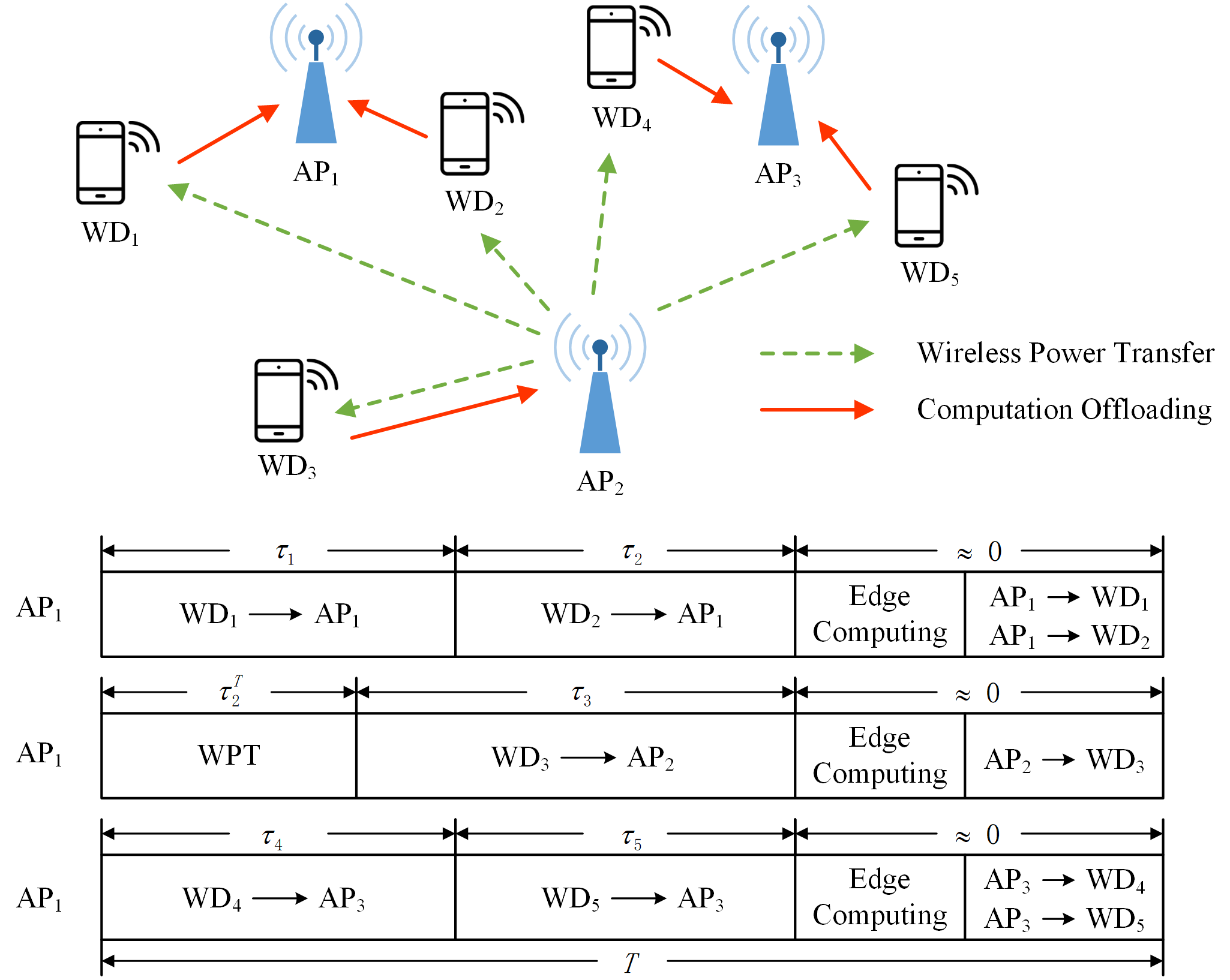}
    \caption{An illustration of the system model and time allocation.}
    \label{fig:wpmec}
\end{figure}

\subsection{Wireless Power Transfer and Energy Harvesting Model}
In scenarios where WDs harvest RF energy from multiple APs simultaneously, they must operate over a broader spectral bandwidth,
which typically exhibits significantly lower conversion efficiency due to circuit design limitations and signal dispersion effects \cite{kwiatkowski2022broadband}. 
To enhance the efficiency of RF energy transfer, we therefore assume that only one AP is selected to broadcast RF energy during each time slot. 
By cyclically activating different APs over time, this strategy also mitigates the ``double near–far'' problem, 
as it ensures that each WD has periodic opportunities to harvest energy from a geographically closer AP.

We use binary variable $a^T_j(t) \in \{0, 1\}$ to indicate whether AP $j$ is selected to broadcast RF energy in time slot $t$.
The corresponding transmission power and time are denoted by $P^T_j(t)$ and $\tau^T_j(t)$, respectively. 
The following constraint ensures that at most one AP conducts energy transmission per time slot:
\begin{equation}
  \sum_{j=1}^{M} a^T_j(t) \leq 1, \quad \forall i\in\mathcal{N}, \forall t\in \mathcal{T}. \label{cons:a^T_j}
\end{equation}
Accordingly, the energy consumption of AP $j$ for WPT during slot $t$ is given by:
\begin{equation}
    E^T_j(t) = a^T_j(t) P^T_j(t) \tau^T_j(t). \label{eq:wpt_energy}
\end{equation}

Following the assumptions in~\cite{mao2019energy, lu2014wireless}, we neglect the energy harvested from ambient noise and adopt a linear energy harvesting model. 
The amount of energy harvested by WD $i$ in time slot $t$ is expressed as:
\begin{equation}
    E^H_i(t) = \sum_{j=1}^{M} \mu_i a^T_j(t) P^T_j(t) h^D_{ij}(t) \tau^T_j(t) \label{eq:energy_harvest}
\end{equation}
where $\mu_j \in (0,1)$ denotes the energy conversion efficiency of WD $i$.
The harvested energy is stored in the WD’s rechargeable battery.
Let $B_i^{\max}$ denote the battery capacity and $B_i(t)$ represent the available energy of WD $i$ at the beginning of slot $t$.
The update rule of $B_i(t)$ is:
\begin{equation*}
    B_i(t+1) = \min \left[ B_i(t) - E^L_i(t) - E^O_i(t) + E^H_i(t), B^{max}_i \right]
\end{equation*}
where $E^L_i(t)$ and $E^O_i(t)$ are the energy consumption for local computing and computation offloading, respectively.
These terms will be further elaborated in the next subsection.
Due to the energy causality constraint, the energy consumption of each WD in any time slot cannot exceed its available battery energy, which leads to:
\begin{equation}
  E^L_i(t) + E^O_i(t) \leq B_i(t), \quad \forall i\in\mathcal{N}, \forall t\in \mathcal{T}. \label{cons:energy_causality}
\end{equation}

\subsection{Computation Scheduling Model}
In each time slot $t$, let $A_i(t)$ denote the amount of computation data arriving at WD $i$. 
Without loss of generality, $A_i(t)$ is assumed to be an independent and identically distributed (i.i.d.) random variable with an 
average arrival rate of $\mathbb{E}[A_i(t)] = \lambda_i$. 
The newly arrived tasks can either be processed locally or offloaded to one of the APs for edge computation.

Let $Q_i(t)$ represent the length of the data queue maintained at WD $i$. The queue dynamics evolve according to
\begin{equation}
  Q_i(t+1) = \max \left[Q_i(t) - D^L_i(t) - D^O_i(t), 0 \right] + A_i(t),
  \label{eq:Q_update}
\end{equation}
where $D^L_i(t)$ and $D^O_i(t)$ denote the amount of data processed locally and offloaded to the APs, respectively.
It is important to note that newly arrived data $A_i(t)$ becomes available only at the end of the current slot and can be processed starting from the next slot. 

\subsubsection{Local Computing}
Following prior studies, we assume that local computing can be performed simultaneously with energy harvesting and computation offloading.
Each WD $i$ employs the Dynamic Voltage and Frequency Scaling (DVFS) technique~\cite{rabaey2003digital} to dynamically adjust its 
CPU operating frequency $f_i(t)$ for energy-efficient computing. 
Let $\tau^L_i(t)$ be the duration of local computing for WD $i$ during time slot $t$. Then the amount of data processed locally can be expressed as
\begin{equation}
    D^L_i(t) = \frac{f_i(t)\tau^L_i(t)}{\phi_i} \label{eq:def_D^L}
\end{equation}
where $\phi_i$ is the number of CPU cycles required to process one bit of computation data.

According to \cite{hu2018wireless}, the energy consumption associated with local computing at WD $i$ is given by
\begin{equation}
    E^L_i(t) = \kappa_i f^3_i(t) \tau^L_i(t) 
    \label{eq:def_E^L}
\end{equation}
where $\kappa_i$ denotes the energy efficiency coefficient of the processor equipped on WD $i$.
By substituting \eqref{eq:def_E^L} into \eqref{eq:def_D^L} we can obtain
\begin{equation*}
    D^L_i(t) = \sqrt[3]{\frac{E^L_i(t){\tau^L_i(t)}^2}{\kappa_i}} \frac{1}{\phi_i}
\end{equation*}
which indicates that $D^L_i(t)$ is a monotonically increasing function of the computation duration $\tau^L_i(t)$ for a given energy budget $E^L_i(t)$.
Therefore, to maximize energy efficiency, it is optimal to allocate the entire slot to local processing, i.e.,
\begin{equation*}
  \tau^L_i(t) = \Delta t, \quad \forall i\in\mathcal{N}, \forall t\in\mathcal{T}.
\end{equation*}

\subsubsection{Computation Offloading}
Let $\tau^O_i(t)$ and $P^O_i(t)$ denote the offloading duration and transmission power of WD $i$ during time slot $t$, respectively. 
Each WD is assumed to communicate with at most one AP within a single time slot. 
We introduce a binary variable $a^O_{ij}(t) \in \{0, 1\}$ to indicate whether 
WD $i$ is associated with AP $j$ for data offloading in slot $t$,
then $a^O_{ij}(t)$ should satisfy
\begin{equation}
  \sum_{j=1}^{M} a^O_{ij}(t) \leq 1, \quad \forall i\in\mathcal{N}, \forall t\in \mathcal{T}. \label{cons:a_ij}
\end{equation}
The amount of computation data offloaded from WD $i$ to the APs in slot $t$ is given by
\begin{equation}
    D^O_i(t) = \sum_{j=1}^{M} \frac{a^O_{ij}(t)B\tau^O_{i}(t)}{v_i} \log_2 \left( 1 + \frac{P^O_i(t) h^U_{ij}(t)}{\sigma_j^2} \right)
    \label{eq:offloaded_data}
\end{equation}
where $B$ denotes the spectral bandwidth, $\sigma_j^2$ represents the noise power at AP $j$, 
and $v_i > 1$ accounts for the communication overhead caused by encryption, 
packet headers, and protocol signaling~\cite{bi2018computation, zhou2020computation}.  
The energy consumption incurred by WD $i$ for computation offloading is expressed as
\begin{equation}
    E^O_i(t) = P^O_i(t) \tau^O_i(t).
    \label{eq:def_offloading_energy}
\end{equation}
Similar to~\cite{wang2017joint}, 
we assume the computation energy consumption at AP $j$ is proportional to the total volume of data received from WDs,
which is given by
\begin{equation*}
    E^C_j(t) = \sum_{i=1}^{N} \eta a^O_{ij}(t) \phi_i D^O_i(t)
\end{equation*}
where $\eta$ represents the energy consumption per CPU cycle of the MEC server at AP $j$.

\subsection{Problem Formulation} \label{subsection:problem_formulation}
To reduce the overall carbon footprint, this paper aims to minimize the total energy consumption incurred in 
processing the computation tasks generated by the WDs. 
This objective is achieved through the joint optimization of WPT, local computing, and computation offloading decisions. 
A similar optimization objective has been investigated in~\cite{ji2018energy, mao2019energy, you2016energy}. 
Based on the system model developed in the preceding subsections, the problem can be formulated as follows:
\begin{align}
  \min_{\bm{a, \tau, P, f}}\quad & \lim_{T\to\infty} \frac{1}{T} \sum_{t=0}^{T-1} \sum_{j=1}^{M} \mathbb{E}\left\{ E^T_j(t) + E^C_j(t) \right\} \label{prob} \\
  s.t.\quad & \eqref{cons:a^T_j}, \eqref{cons:energy_causality}, \eqref{cons:a_ij} \quad \tag{\ref{prob}{a}} \label{cons:before_constraints} \\
            & a^T_j(t)\tau^T_j(t) + \sum_{i=1}^{N} a^O_{ij}(t)\tau^O_i(t) \leq \Delta t, \notag \\
            & \qquad \qquad \qquad \qquad \qquad \ \ \forall j\in\mathcal{M}, \forall t\in\mathcal{T} \tag{\ref{prob}{b}} \label{cons:time_allocation} \\
            & a^T_j(t), a^O_{ij}(t) \in \{ 0, 1\},  \notag \\
            & \qquad \qquad \qquad \quad \forall i\in\mathcal{N}, \forall j\in\mathcal{M}, \forall t\in\mathcal{T} \tag{\ref{prob}{c}} \label{cons:var_a} \\
            & 0 \leq \tau^T_j(t) \leq \Delta t, 0 \leq P^T_j(t) \leq P^{T, max}_j, \notag \\
            & \qquad \qquad \qquad \qquad \qquad \ \forall j\in\mathcal{M}, \forall t\in\mathcal{T} \tag{\ref{prob}{d}} \label{cons:var_wpt} \\
            & 0 \leq \tau^O_{i}(t) \leq \Delta t, 0 \leq P^O_i(t) \leq P^{O,max}_i, \notag \\ 
            & \qquad \qquad \qquad \qquad \qquad \ \ \forall i\in\mathcal{N}, \forall t\in\mathcal{T} \tag{\ref{prob}{e}} \label{cons:var_offloading} \\
            & 0 \leq f_i(t) \leq f^{max}_i, \quad \forall i\in\mathcal{N}, \forall t\in\mathcal{T} \tag{\ref{prob}{f}} \label{cons:var_local} \\
            & \widebar{Q}_i < \infty, \quad \forall i\in\mathcal{N} \tag{\ref{prob}{g}} \label{cons:stability}
\end{align}
where $\bm{a, \tau, P, f}$ represents the sets of decision variables, 
$P^{T,max}_j$ denotes the maximum WPT transmission power of AP $j$,
$P^{O,max}_i$ and $f^{max}_i$ are the maximum offloading power and computational capability of WD $i$,
$\widebar{Q}_i = \lim_{T\to\infty} \frac{1}{T} \sum_{t=0}^{T-1} Q_i(t)$ is the time average length of $Q_i(t)$.

The objective reflects the long-term average energy consumption of all APs,
which also reflects the total system-wide energy consumption since WDs are powered by energy harvested from the APs.
Constraint~\eqref{cons:before_constraints} incorporates the system constraints defined in the previous subsections.
Constraint~\eqref{cons:time_allocation} enforces feasible time allocation under the TDMA protocol.
Constraints~\eqref{cons:var_a}-\eqref{cons:var_local} specify the allowable domains of the decision variables.
Finally, constraint~\eqref{cons:stability} 
guarantees the stability of the task queues at all WDs, ensuring sustainable network operation over time.

For scenarios in which WDs do not possess SWIPT capabilities, wireless communication cannot be conducted during WPT. 
Consequently, the computation offloading time of \emph{all} APs, rather than only the AP performing energy transfer, 
must exclude the duration of WPT. Accordingly, constraint~\eqref{cons:time_allocation} can be replaced by
\begin{equation*}
\sum_{j'=1}^{M} a^T_{j'} \tau^T_{j'}(t) + \sum_{i=1}^{N} a^O_{ij}(t)\tau^O_i(t) \leq T, \quad \forall j,\ \forall t.
\end{equation*}
It is straightforward to verify that the algorithm proposed in the next section can be readily adapted to this setting.

\section{An Online Algorithm for Wireless Power Transfer and Computation Scheduling} \label{section:algorithm design}
In this section, we develop an online algorithm for the joint optimization of WPT and computation scheduling
based on the Lyapunov optimization framework. 
To efficiently address the inherent coupling among decision variables, 
a relax-then-adjust strategy is introduced to decompose the original 
optimization problem into a series of more tractable subproblems. 
This decomposition significantly reduces computational complexity while maintaining 
near-optimal performance in dynamic network environments.

\subsection{Online Scheduling with Lyapunov Optimization}
Lyapunov optimization provides an effective online optimization framework that is particularly well-suited for queueing systems. 
It enables dynamic decision-making at each time slot to jointly minimize both the system objective and the queue lengths. 
In the proposed system, two types of queues are considered: the data queue $Q_i(t)$ and the energy queue $B_i(t)$. 
However, unlike typical queue stabilization problems, our objective is to maintain a sufficient level of stored energy, 
i.e., to \emph{increase} rather than reduce the energy queue.

To facilitate the application of Lyapunov optimization, we introduce an auxiliary battery deficit queue $B^-_i(t)$, defined as
\begin{equation*}
    B^-_i(t) = B^{max}_i - B_i(t),
\end{equation*}
which represents the energy shortfall relative to the maximum battery capacity. 
According to the energy causality constraint~\eqref{cons:energy_causality}, 
the total energy consumed by WD $i$ in any time slot cannot exceed the energy currently stored in its battery. 
Consequently, $B^-_i(t)$ always lies within the interval $[0, B^{max}_i]$.

Let $\bm{\Theta}(t) = [ \bm{Q}(t), \bm{B}^-(t) ]$ be the combined queue vector, where $\bm{Q}(t) = (Q_1(t), Q_2(t), \dots, Q_N(t))$ and
$\bm{B}^-(t) = (B^-_1(t), B^-_2(t), \dots, B^-_N(t))$.
Following standard Lyapunov optimization principles,
we define the quadratic Lyapunov function as
\begin{equation}
    L(t) = \frac{1}{2} \sum_{i=1}^{N} \left[ Q_i(t)^2 + B^-_i(t)^2 \right].
    \label{eq:lyapunov_function}
\end{equation}
To characterize the change of queue lengths, the one-step conditional Lyapunov drift is defined as
\begin{equation*}
    \Delta L(t) = \mathbb{E} \left\{ L(t+1) - L(t) | \bm{\Theta}(t) \right\}.
\end{equation*}
To jointly balance queue stability and energy efficiency, we combine the Lyapunov drift with the 
system energy consumption objective, leading to the drift-plus-penalty function:
\begin{equation*}
    \Delta_V L(t) = \Delta L(t) + V\mathbb{E} \left\{ \sum_{j=1}^{M} \left( E^T_j(t) + E^C_j(t) \right) | \bm{\Theta}(t) \right\}
\end{equation*}
where $V > 0$ is a tunable control parameter that determines the trade-off between system energy consumption and queue backlog.
The following lemma provides an upper bound of $\Delta_V L(t)$:
\begin{lemma}
    For any time slot $t$ and queue state $\bm{\Theta}(t)$, the drift-plus-penalty term is upper-bounded by
    \begin{align}
        \Delta_V L(t) \leq & C_1 - \sum_{i=1}^{N} Q_i(t) \mathbb{E} \left\{ D^L_i(t) + D^O_i(t) - A_i(t) | \bm{\Theta}(t) \right\} \notag \\
        &- \sum_{i=1}^{N} B^-_i(t) \mathbb{E} \left\{ E^H_i(t) - E^L_i(t) - E^O_i(t) | \bm{\Theta}(t) \right\} \notag \\
        &+ V\mathbb{E} \left\{ \sum_{j=1}^{M} \left( E^T_j(t) + E^C_j(t) \right) | \bm{\Theta}(t) \right\} \label{eq:dpp_bound}
    \end{align}
    where $C_1$ is a finite positive constant.
    \label{lemma:dpp_bound}
\end{lemma}

The proof follows standard procedures~\cite{neely2010stochastic} and is omitted for brevity. 
Based on the Lyapunov optimization framework, the control policy at each time slot 
is to opportunistically minimize the right-hand side of~\eqref{eq:dpp_bound},
which leads to the following per-slot optimization problem:
\begin{align}
    \min_{\bm{a, \tau, P, f}} \quad & \sum_{j=1}^{M} V \left[ E^T_j(t) + E^C_j(t) \right] \notag \\
    &- \sum_{i=1}^{N} Q_i(t) \left[ D^L_i(t) + D^O_i(t) - A_i(t) \right] \notag \\
    &- \sum_{i=1}^{N} B^-_i(t) \left[ E^H_i(t) - E^L_i(t) - E^O_i(t) \right] \label{prob:lya} \\
    s.t. \quad &\eqref{cons:before_constraints}-\eqref{cons:var_local} \tag{\ref{prob:lya}{a}} \label{cons:lya}
\end{align}

However, directly solving problem~\eqref{prob:lya} is computationally challenging due to the 
strong coupling among decision variables. 
For instance, the control variables for local computing and computation offloading are interdependent 
through constraint~\eqref{cons:energy_causality},
while the time allocation for WPT and wireless communication is jointly 
constrained by~\eqref{cons:time_allocation}. 
To efficiently address this issue, we introduce a \emph{relax-then-adjust} approach that decomposes 
the original problem into several independent subproblems, thus significantly reducing computational complexity. 
The details of this algorithm design are presented in the following subsections.

\subsection{Optimality Condition} 
The main idea of the proposed technique is to first relax the coupling constraints among different types 
of control variables so that the original problem~\eqref{prob:lya} can be decomposed into a set of more tractable subproblems. 
After obtaining the preliminary solutions to these subproblems, 
the control decisions are subsequently adjusted to ensure feasibility with respect to the original problem formulation. 
To obtain high-quality feasible solutions, we identify a necessary condition for optimality, 
which is then employed to guide the adjustment process. 

In the considered system, the optimization objective is to minimize the total energy consumption required 
to process the computation tasks of all WDs. 
For any given amount of arrived workload, this objective is equivalent to maximizing the overall energy efficiency, 
defined as the ratio between the total amount of completed computation tasks and the corresponding energy consumption. 
Motivated by this observation, we introduce the concept of \emph{marginal energy efficiency} for both 
local computing and computation offloading. 
These quantities, denoted by $\epsilon^L_i(t)$ and $\epsilon^O_i(t)$, respectively, are defined as
\begin{align}
    \epsilon^L_i(t) &= \frac{\partial E^L_i(t)}{\partial D^L_i(t)} = 3\kappa_i \phi_i f^2_i(t) \notag \\
    \epsilon^O_i(t) &= \frac{\partial \left( E^O_i(t) + \eta \phi_i D^O_i(t) \right)}{\partial D^O_i(t)} \notag \\
    &= \sum_{j=1}^{M} \frac{a^O_{ij}(t)v_i \ln 2}{B} \left( \frac{\sigma^2_j}{h^U_{ij}(t)} + P^O_i(t) \right) + \eta\phi_i \label{def:offloading_ee}
\end{align}
where the time allocation $\tau^O_i(t)$ is assumed to be given. 
The term $\epsilon^L_i(t)$ represents the incremental energy required to locally process one 
additional bit of data, whereas $\epsilon^O_i(t)$ reflects the corresponding incremental energy 
for offloading computation to an AP, including both transmission and edge-processing costs.
The following theorem establishes a necessary condition for optimality in the energy-efficient scheduling problem.
\begin{theorem}
  \label{theorem:optimality_condition}
  At the optimal solution, the marginal energy efficiencies of local computing and computation offloading of 
  WD $i$ are equal if (i) its offloading time is non-zero, and (ii) $f_i(t)$ and $P^O_i(t)$ are not at the boundary of their domain, i.e., 
  $\epsilon^L_i(t) = \epsilon^O_i(t)$, if $\tau^O_i(t) > 0$, $f_i(t) \in (0, f^{max}_i)$, and $P^O_i \in (0, P^{O,max}_i)$.
\end{theorem}
\begin{IEEEproof}
  The proof proceeds by contradiction. 
  Without loss of generality, assume that (i) WD $i$ has sufficient queued data for processing, 
  and (ii) there exists an optimal solution 
  $\bm{x}^* = (\bm{a}^*, \bm{\tau}^*, \bm{P}^*, \bm{f}^*)$ such that $\epsilon^L_i(t) > \epsilon^O_i(t)$ while
  $\tau^{O,*}_i(t) > 0$, $f^*_i(t) \in (0, f^{max}_i)$, and $P^{O,*}_i(t) \in (0, P^{O,max}_i)$. 
  We construct a new feasible solution $\bm{x}'$ by slightly reducing the CPU frequency of WD $i$, 
  i.e., $f'_i(t) = f^*_i(t) - \delta$, for a sufficiently small $\delta > 0$. 
  This adjustment saves an amount of local computing energy given by 
  \begin{equation*}
    \Delta E^L_i(t) = \kappa_i \Delta t \left[ {f^*_i}^3(t) - {f'_i}^3(t) \right].
  \end{equation*}
  Since $\tau^{O,*}_i(t) > 0$, we can then increase the offloading power to 
  \begin{equation*}
    P^{O,'}_i(t) = P^{O,*}_i(t) + \frac{\Delta E^L_i(t)}{\tau^{O,*}_i(t)}, 
  \end{equation*}
  while keeping all other control variables unchanged. 
  Intuitively, this modification reallocates a small portion of the energy budget from local computing to computation offloading.
  Note that $P^{O,*}_i(t) < P^{O,max}_i$, hence $P^{O,'}_i(t)$ lies in the feasible domain for sufficiently small $\delta$.
  Since $\epsilon^L_i(t) > \epsilon^O_i(t)$ and both quantities are continuous functions of $f_i(t)$ and $P^O_i(t)$, 
  it follows that for a sufficiently small $\delta$, $\bm{x}'$ processes more computation data than
  $\bm{x}^*$ with the same energy consumption, contradicting the assumption that $\bm{x}^*$ is optimal. 
  Hence, the optimal solution must satisfy $\epsilon^L_i(t) = \epsilon^O_i(t)$.
\end{IEEEproof}

This optimality condition serves as a guiding principle for the adjustment of control variables after 
obtaining the preliminary solutions of subproblems.
The detailed algorithm design is presented in the next subsection.

\subsection{Algorithm Design}
As discussed above, our strategy is to first to relax the coupling constraints~\eqref{cons:energy_causality} and \eqref{cons:time_allocation}, 
so that problem~\eqref{prob:lya} can be decomposed into a set of independent subproblems that are easier to solve.
Once the optimal solutions to these subproblems are obtained, the control variables are adjusted according 
to the established optimality condition to ensure that the final decisions are 
feasible with respect to the original problem constraints. 

\subsubsection{WPT Subproblem} \label{section:wpt}
By substituting~\eqref{eq:wpt_energy} and~\eqref{eq:energy_harvest} into the drift-plus-penalty bound~\eqref{eq:dpp_bound} 
and grouping the WPT-related terms together, we can obtain the following WPT subproblem:
\begin{align}
  \min_{\bm{a}^T(t), \bm{P}^T(t), \bm{\tau}^T(t)}\quad & \sum_{j=1}^{M} \left( V - \sum_{i=1}^{N} B^-_i(t)\mu_i h^D_{ij}(t) \right) \notag \\
                                                       & \qquad \qquad \qquad \times a^T_j(t) P^T_j(t) \tau^T_j(t)  \label{obj:wpt} \\
    s.t. \quad & \eqref{cons:a^T_j}, \eqref{cons:var_wpt} \notag
\end{align}
Note that we have temporarily relaxed the time allocation constraint \eqref{cons:time_allocation}.
Let $c^T_j(t) = V - \sum_{i=1}^{N} B^-_i(t)\mu_i h^D_{ij}(t)$ be the coefficient of $a^T_j(t) P^T_j(t) \tau^T_j(t)$.
Since at most one AP can be selected to broadcast RF energy in each time slot,
the WPT subproblem can be easily solved by selecting the AP with the smallest negative $c^T_j(t)P^{T,max}_{j}$.
Specifically, if $c^T_j(t) \geq 0$ for all $j$, then the optimal value is obtained by setting all $a^T_j(t)$ to $0$,
indicating that no RF energy transfer occurs during slot $t$.
Otherwise, find $j^* = \argmin c^T_j(t) P^{T, max}_j$ and the corresponding optimal decisions are
$a^T_{j^*}(t) = 1, P^T_{j^*}(t) = P^{T,max}_{j^*}, \tau^T_{j^*}(t) = \Delta t$.

\subsubsection{Local Computing Subproblem} \label{section:local_computation}
In the local computing process, the only controllable variable is the CPU frequency of each WD. 
By relaxing the coupling constraint~\eqref{cons:energy_causality},
the local computing subproblem can be formulated as
\begin{align}
  \min_{\bm{f}(t)}\quad &\sum_{i=1}^{N} B^-_i(t) \kappa_i f^3_i(t) \Delta t - \sum_{i=1}^{N} Q_i(t) \frac{f_i(t) \Delta t}{\phi_i} \label{subprob:local_computing} \\
  s.t. \quad & \eqref{cons:var_local}  \notag \\
             & E^L_i(t) \leq B_i(t), \quad \forall i\in\mathcal{N} \tag{\ref{subprob:local_computing}a} \label{cons:local_computing_energy}
\end{align}
where we have relaxed the energy causality constraint~\eqref{cons:energy_causality} to \eqref{cons:local_computing_energy}
so that it only restricts the energy consumption due to local computing.
By substituting the expression of $E^L_i(t)$, given in \eqref{eq:def_E^L}, into \eqref{cons:local_computing_energy},
we have 
\begin{equation*}
  f_i(t) \leq \sqrt[\raisebox{3pt}{\scriptsize 3}]{\frac{B_i(t)}{\kappa_i \Delta t}}.
\end{equation*}
Define the upper bound of $f_i(t)$ as 
\begin{equation*}
  f^{ub}_i(t) = \min \left\{ \sqrt[\raisebox{3pt}{\scriptsize 3}]{\frac{B_i(t)}{\kappa_i \Delta t}}, f^{max}_i \right\}.
\end{equation*}
Since the objective function~\eqref{subprob:local_computing} is separable with respect to $f_i(t)$, 
the above problem can be further decomposed into $N$ independent per-device subproblems, each expressed as
\begin{equation}
  \min_{f_i(t)\in [0, f^{ub}_i(t)]} \quad B^-_i(t) \kappa_i f^3_i(t) \Delta t - Q_i(t) \frac{f_i(t) \Delta t}{\phi_i}. \label{subsubproblem:local}
\end{equation}
Each subproblem~\eqref{subsubproblem:local} is a univariate optimization problem and can be solved analytically. 
The optimal CPU frequency is attained either at the stationary point of the objective function or at one of the boundary values. 
By setting the derivative with respect to $f_i(t)$ to zero and applying the upper bound constraint, 
the optimal CPU frequency for WD $i$ is given by
\begin{equation}
    f_i(t) = \min \left\{ f^{ub}_i(t), \sqrt{\frac{Q_i(t)}{3\kappa_i \phi_i B^-_i(t)}} \right\}.
    \label{eq:optimal_f}
\end{equation}

\subsubsection{Computation Offloading Subproblem} \label{section:offloading_subproblem}
The remaining control variables are associated with the computation offloading process. 
The corresponding optimization subproblem can be formulated as:
\begin{align}
  \min_{\bm{a}^O(t), \bm{P}^O(t), \bm{\tau}^O(t)} & \sum_{j=1}^{M} \sum_{i=1}^{N} V\eta a^O_{ij}(t)\phi_i D^O_i(t) - \sum_{i=1}^{N} Q_i(t)D^O_i(t) \notag \\
    & + \sum_{i=1}^{N} B^-_i(t) P^O_i(t) \tau^O_i(t) \label{obj:offloading} \\
    s.t.\quad & \eqref{cons:a_ij}, \eqref{cons:var_offloading} \notag \\
              & \sum_{i=1}^{N} a^O_{ij}(t)\tau^O_i(t) \leq T, \quad \forall j\in\mathcal{M} \tag{\ref{obj:offloading}a}\label{cons:relax_time_allocation}  \\
              & E^O_i(t) \leq B_i(t), \quad \forall i\in\mathcal{N} \tag{\ref{obj:offloading}b} \label{cons:offloading_energy}
\end{align}
In this formulation, the original time allocation constraint \eqref{cons:time_allocation} is 
relaxed to~\eqref{cons:relax_time_allocation} in order to simplify the optimization process. 
Similar to the local computing subproblem, the energy causality constraint~\eqref{cons:energy_causality} is relaxed to~\eqref{cons:offloading_energy}, 
restricting the instantaneous energy consumption for offloading within the available battery energy.  
By substituting~\eqref{eq:def_offloading_energy} into~\eqref{cons:offloading_energy}, we obtain the equivalent upper bound for the transmission power:
$P^O_i(t) \leq \frac{B_i(t)}{\tau^O_i(t)}$.
As will be demonstrated later, the optimal time allocation $\tau^O_i(t)$ is either $\Delta t$ or $0$, 
implying that the above condition can be equivalently expressed as
$P^O_i(t) \leq \frac{B_i(t)}{\Delta t}$.
Define the upper bound of the transmission power $P^O_i(t)$ as
$P^{O,ub}_i(t) = \min \left\{ P^{O,max}_i, \frac{B_i(t)}{\Delta t} \right\}$,
then constraint \eqref{cons:offloading_energy} can be combined with the domain of $P^O_i(t)$ to obtain the unified constraint:
\begin{equation*}
  0 \leq P^O_i(t) \leq P^{O,ub}_i(t), \quad \forall i\in\mathcal{N}.
\end{equation*}

The resulting mixed-integer programming (MIP) problem is non-convex and, in general, does not admit an efficient closed-form solution.
Nevertheless, through detailed theoretical analysis, it can be shown that problem~\eqref{obj:offloading} can be 
equivalently transformed into an \emph{Assignment Problem}, which can be solved optimally in polynomial time using established algorithms. 
The detailed derivation and algorithmic procedure are presented below.

\textbf{Step 1: Decomposition by Association Decisions.}
We first assume that the association between WDs and APs is given. 
Let $\mathcal{N}_j(t) = \{ i \,|\, a^O_{ij}(t) = 1 \}$ denote the set of WDs associated with AP $j$ during time slot $t$. 
Given the association decisions, the optimization problem~\eqref{obj:offloading} becomes separable across APs. 
Thus, for each AP $j$, the corresponding offloading subproblem can be expressed as:
\begin{align}
  \min_{P^O_i(t), \tau^O_i(t)}\quad & \sum_{i\in\mathcal{N}_j} V\eta \phi_i D^O_i(t) - \sum_{i\in\mathcal{N}_j} Q_i(t) D^O_i(t) \notag \\
                                    & + \sum_{i\in\mathcal{N}_j} B^-_i(t) P^O_i(t) \tau^O_i(t) \label{subproblem:offloading_obj} \\
  s.t. \quad & \sum_{i\in\mathcal{N}_j} \tau^O_i(t) \leq \Delta t \notag \\
             & 0 \leq \tau^O_i(t) \leq \Delta t, \quad \forall i\in\mathcal{N}_j \notag \\
             & 0 \leq P^O_i(t) \leq P^{O,ub}_i(t), \quad \forall i\in\mathcal{N}_j \notag
\end{align}
By substituting the expression of $D^O_i(t)$ into~\eqref{subproblem:offloading_obj}, the objective can be rewritten as
$\sum_{i\in\mathcal{N}_j} c^O_{ij}(t) \tau^O_i(t)$,
where the coefficient $c^O_{ij}(t)$ is defined as:
\begin{align}
  c^O_{ij}(t) = & \frac{\left( V\eta \phi_i - Q_i(t) \right) B}{v_i} \log_2 \left( 1 + \frac{P^O_i(t) h^U_{ij}(t)}{\sigma^2_j} \right) \notag \\
                & + B^-_i(t) P^O_i(t). \label{eq:c_O_ij}
\end{align}

\textbf{Step 2: Optimal Transmission Power.}
We first derive the optimal transmission power. 
Given that $\tau^O_i(t)\geq 0$, to minimize the objective function,
the optimal value of $P^O_i(t)$ must minimize $c^O_{ij}(t)$ for any feasible $\tau^O_i(t)$. 
Therefore, the candidate solutions lie either at the boundary of the feasible power region or at a stationary point of $c^O_{ij}(t)$. 
By setting the first-order derivative of $c^O_{ij}(t)$ with respect to $P^O_i(t)$ to zero,
the analytical expression of the stationary point can be obtained as follows:
\begin{equation*}
  \widehat{P}^O_i(t) = \frac{( Q_i(t) - V\phi_i\eta )B}{B^-_i(t) v_i \ln 2} - \frac{\sigma^2_{j}}{h^U_{ij}(t)}.
\end{equation*}
Consequently, the corresponding optimal transmission power is
\begin{equation}
  P^O_i(t) = \max \left\{ 0,  \min \left\{P^{O,ub}_i(t), \widehat{P}^O_i(t) \right\} \right\}
  \label{eq:P_O_i}
\end{equation}

\textbf{Step 3: Optimal Time Allocation.}
Next, we show how to obtain the optimal time allocation.
Similar to the WPT subproblem,
to minimize $\sum_{i\in\mathcal{N}_j} c^O_{ij}(t) \tau^O_i(t)$,
the entire transmission duration $\Delta t$ should be assigned to the WD with the smallest negative coefficient $c^O_{ij}(t)$.
Specifically, let $i^* = \argmin_{i\in\mathcal{N}_j} c^O_{ij}(t)$.
If $c^O_{i^*j}(t) < 0$, then we set $\tau^O_{i^*}(t) = \Delta t$ and $\tau^O_i(t) = 0$ for all $i \neq i^*$.
Otherwise, set $\tau^O_i = 0$ for all $i \in \mathcal{N}_j$.

\textbf{Step 4: Association Optimization.}
Based on the above analysis, we can finally optimize the association between WDs and APs.
As discussed previously, each AP will allocate its entire transmission time to the WD that yields the minimum negative value of $c^O_{ij}(t)$. 
Since the remaining WDs have no transmission time, their contribution to the objective function becomes zero. 
Consequently, these devices can be treated as unassociated.

Therefore, the association optimization reduces to selecting exactly one WD for each AP such that the total cost $\sum c^O_{ij}(t)\Delta t$ is minimized. 
This problem can be equivalently formulated as a classical Assignment Problem. 
In the corresponding bipartite graph representation, one partition consists of WDs and the other consists of APs. 
The weight of the edge between WD $i$ and AP $j$ is defined as
\begin{equation}
  w_{ij}(t) = 
  \begin{cases}
    c^O_{ij}(t) \Delta t, & \text{if } c^O_{ij}(t) < 0, \\
    0, & \text{otherwise.}
  \end{cases}
  \label{eq:w_ij}
\end{equation}
Note that we set the weight $w_{ij}(t)$ to $0$ if the corresponding $c^O_{ij}(t) \geq 0$, as no transmission time will be allocated 
to WDs with non-negative coefficient, hence yeilding zero contribution to the overall objective.

The assignment problem can be solved efficiently using the Hungarian Algorithm with computational complexity $O(N^3)$, under the assumption $N > M$ \cite{kuhn1955hungarian}. 
Solving this produces the optimal association matrix $a^O_{ij}(t)$, where each AP is associated with exactly one WD. 
Subsequently, the optimal transmission time and power variables can be updated based on the results obtained above.

\subsubsection{Improved Algorithm for Computation Offloading} \label{section:improved_offloading}
According to the proposed \emph{relax-then-adjust} strategy, the adjustment of control variables should be performed \emph{after} 
deriving the preliminary solutions of all subproblems. 
However, in the computation offloading subproblem, the optimal association decision depends on the coefficients $c^O_{ij}(t)$, 
which are functions of the transmission power $P^O_i(t)$ determined from~\eqref{eq:P_O_i}. 
If the resulting solution violates the energy causality constraint, the values of $P^O_i(t)$ and the CPU frequency $f_i(t)$ must be 
adjusted using the optimality condition in Theorem~\ref{theorem:optimality_condition}. 
Nevertheless, modifying $P^O_i(t)$ after solving the offloading subproblem 
implies that the previously determined association decisions may no longer retain optimality.

To circumvent this issue, we perform the adjustment of $f_i(t)$ and $P^O_i(t)$ \emph{prior} to determining the association and offloading duration. 
Specifically, upon obtaining the initial value of $P^O_i(t)$ in Step~2, 
we immediately verify whether the energy causality constraint~\eqref{cons:energy_causality} is satisfied. 
If so, the algorithm proceeds without modification. 
Otherwise, both $f_i(t)$ and $P^O_i(t)$ are updated such that $E^L_i(t) + E^O_i(t) = B_i(t)$,
while simultaneously enforcing the optimality condition $\epsilon^L_i(t) = \epsilon^O_i(t)$.
Combining these two equalities leads to a closed-form characterization of the optimal $f_i(t)$, 
obtained by solving the following univariate cubic equation:
\begin{align}
  \frac{v_i \ln 2}{B} \kappa_i f^3_i(t) + & 3 \kappa_i \phi_i f^2_i(t) \notag \\
                                          & = \frac{v_i \ln 2}{B} \left( \frac{\sigma^2_j}{h^U_{ij}(t)} + \frac{B_i(t)}{\Delta t} \right) + \eta \phi_i.
  \label{eq:cubic_f_i}
\end{align}
Let $f_i^*(t)$ denote the meaningful solution of~\eqref{eq:cubic_f_i}. We then update the transmission power as:
\begin{equation}
  P^O_i(t) = \frac{B_i(t)}{\Delta t} - \kappa_i {f^*_i}^3(t).
  \label{eq:update_P}
\end{equation}
With the updated $P^O_i(t)$, the coefficients $c^O_{ij}(t)$ can be recalculated.

Next, the updated coefficients are employed in Step~4 to optimize the association decisions and to determine the offloading time. 
For WDs assigned a nonzero offloading duration, their CPU frequencies are set to the adjusted values, i.e., $f_i(t) = f_i^*(t)$. 
For the remaining WDs, $f_i(t)$ is kept to the original optimal value obtained from~\eqref{eq:optimal_f}.
The complete procedure is summarized in Algorithm~\ref{alg:offloading}.

\begin{algorithm}[t]
  \caption{Algorithm for the computation offloading subproblem}
  \label{alg:offloading}
    \SetKwInOut{KwIn}{Input}

    \KwIn{$\bm{f}(t)$ calculated from \eqref{eq:optimal_f}}

    Initialize: $\bm{a}^O(t), \bm{P}^O(t), \bm{\tau}^O(t), \bm{y}(t) \leftarrow \bm{0}$\;
    \For{$j \in \mathcal{M}$}{
      \For{$i \in \mathcal{N}$}{
        Calculate $P^O_i(t)$ according to \eqref{eq:P_O_i}\;
        \If{$E^L_i(t) + E^O_i(t) > B_i(t)$}
        {
          $y_{ij}(t) \leftarrow 1$\;
          Obtain $f^*_i(t)$ by solving \eqref{eq:cubic_f_i}\;
          Calculate $P^O_i(t)$ according to \eqref{eq:update_P}\;
        }
        Calculate $c^O_{ij}(t)$ according to \eqref{eq:c_O_ij}\;
       }
    }
    Calculate $w_{ij}(t)$ according to \eqref{eq:w_ij}\;
    Obtain the optimal association $a^O_{ij}(t)$ using the Hungarian Algorithm\;
    \For{$j \in \mathcal{M}$}{
      Find $i^*$ such that $a^O_{i^*j} = 1$\;
      \If{$c^O_{i^*j}(t) < 0$}
      {
        Set $\tau^O_{i^*} \leftarrow \Delta t$\;
        \eIf{$y_{i^* j}(t) = 1$}
        {
          Obtain $f^*_{i^*}(t)$ by solving \eqref{eq:cubic_f_i}\;
          $f_{i^*}(t) \leftarrow f^*_{i^*}(t)$\;
          Update $P^O_{i^*}(t)$ according to \eqref{eq:update_P}\;
        }
        {
          Update $P^O_{i^*}(t)$ according to \eqref{eq:P_O_i}\;
        }
      }
    }
    \KwRet{$\bm{a}^O(t), \bm{P}^O(t), \bm{\tau}^O(t), \bm{f}(t)$}.
\end{algorithm}

\subsubsection{Adjustment for the Time Allocation Constraint} 
The improved algorithm described in the previous subsection ensures compliance with the energy causality constraint. 
However, the relaxed time allocation constraint~\eqref{cons:relax_time_allocation} may still violate the original constraint~\eqref{cons:time_allocation}. 
Recall that constraint~\eqref{cons:relax_time_allocation} is obtained by omitting the WPT duration in the left-hand side of~\eqref{cons:time_allocation}, 
meaning that a violation can occur at the AP that broadcasts RF energy.

Let $j^*$ denote the AP performing WPT during time slot $t$. According to the result derived in Section~\ref{section:wpt}, 
we have $\tau^T_{j^*}(t) = \Delta t$.
Moreover, as analyzed in Section~\ref{section:offloading_subproblem}, at most one WD can be associated with AP $j^*$ for 
computation offloading in the same time slot. 
Denote this WD as $i^*$. 
From the allocation strategy described in Section~\ref{section:improved_offloading}, we similarly have $\tau^O_{i^*}(t) = \Delta t$.

To ensure satisfaction of the original constraint~\eqref{cons:time_allocation}, 
we have to adjust the duration of these two processes so that $\tau^T_{j^*}(t) + \tau^O_{i^*}(t) \leq \Delta t$.
To determine the optimal time allocation, we compare the contributions of the two processes to the objective function. 
The coefficients of $\tau^T_{j^*}(t)$ and $\tau^O_{i^*}(t)$ in the objective are given by $c^T_{j^*}(t) P^{T,\max}_{j^*}$ and $c^O_{i^* j^*}(t)$, respectively. 
Therefore, to minimize the objective, the entire slot duration should be allocated to the process with the smaller coefficient,
i.e., let
\begin{equation*}
  \left( \tau^T_{j^*}(t), \tau^O_{i^*}(t) \right) =
  \begin{cases}
    (\Delta t, 0), &\text{if } c^T_{j^*}(t) P^{T,max}_{j^*} < c^O_{i^* j^*}(t) \\
    (0, \Delta t), &\text{otherwise }
  \end{cases}.
\end{equation*}
The complete procedure is summarized in Algorithm \ref{alg:time_allocation}.

The computational complexity of the complete algorithm is analyzed as follows. 
The time complexities of the WPT subproblem and the local computing subproblem are $\mathcal{O}(MN)$ and $\mathcal{O}(N)$, respectively. 
The computation offloading subproblem is solved using Algorithm~\ref{alg:offloading}, whose complexity is dominated by employing the Hungarian algorithm
in line 14, resulting in a time complexity of $\mathcal{O}(N^3)$ \cite{kuhn1955hungarian}.
In addition, the adjustment procedure in Algorithm~\ref{alg:time_allocation} incurs a time complexity of $\mathcal{O}(N)$. 
Consequently, the overall time complexity of the proposed single-slot algorithm is $\mathcal{O}(N^3)$.

\begin{algorithm}[t]
  \caption{Adjustment of time allocation} 
  \label{alg:time_allocation}
    \SetKwInOut{KwIn}{Input}


    \KwIn{Preliminary time allocation: $\bm{\tau}(t)$}
    Find $j^*$ such that $a^T_{j^*}(t) = 1$\;
    \For{$i \in \mathcal{N}$}{
      \If{$a^O_{ij^*}(t) = 1$ and $\tau^T_{j^*}(t) + \tau^O_i(t) > \Delta t$}
      {
        \eIf{$c^T_{j^*}(t) P^{T,max}_{j^*} < c^O_{ij^*}(t)$}
        {
          $\tau^T_{j^*}(t) \leftarrow \Delta t, \tau^O_{i^*}(t) \leftarrow 0$\;
        }
        {
          $\tau^T_{j^*}(t) \leftarrow 0, \tau^O_{i^*}(t) \leftarrow \Delta t$\;
        }
      }
    }
    \KwRet{ $\bm{\tau}(t)$ }.
\end{algorithm}

\subsection{Performance Analysis}
In this subsection, we analyze the performance of the proposed algorithm by characterizing the gap between its achieved objective value and the optimal one, 
as well as the inherent trade-off between energy efficiency and average delay. 
Let $P$ denote the original optimization problem~\eqref{prob:lya}, 
and let $P'$ denote its relaxed counterpart obtained by replacing constraint~\eqref{cons:time_allocation} with~\eqref{cons:relax_time_allocation}. 
As established earlier, the preliminary solution produced by Algorithm~\ref{alg:offloading}, 
together with the WPT decision derived in Section~\ref{section:wpt}, corresponds to the optimal solution of $P'$. 
The corresponding objective value is denoted by $z'$.
Define $z^*$ and~$z$ as the optimal objective value of $P$ and the objective value obtained by our final algorithm, respectively. 
We demonstrate that the performance gap $z - z^*$ is bounded by a finite constant.

Let $h^{D,\max}_{i}$ be an upper bound of the downlink channel gain $h^D_{ij}(t)$, 
and define $P^{T,\max} = \max_{j\in\mathcal{M}} P^{T,\max}_j$ as the maximum WPT transmission power among all APs. Then the following result holds:
\begin{lemma}
  The optimality gap between the proposed algorithm and the global optimum is bounded by
  \begin{equation*}
    z - z^* \leq C_2,
  \end{equation*}
  where 
  \begin{equation*}
    C_2 = \left(  \sum_{i=1}^N B^{max}_i \mu_i h^{D,max}_{i} - V \right) P^{T,max} \Delta t.
  \end{equation*}
  \label{lemma:gap}
\end{lemma}
\begin{IEEEproof}
  Since $P'$ is obtained by relaxing $P$, it follows that $z' \leq z^*$. The difference between $z$ and $z'$ arises solely 
  from adjustment operations in Algorithm~\ref{alg:time_allocation}. 
  Specifically, if $c^T_{j^*}(t) P^{T,\max}_{j^*} < c^O_{i^* j^*}(t)$, then:
  \[
  z - z' = -c^O_{i^* j^*}(t) \Delta t,
  \]
  otherwise:
  \[
  z - z' = -c^T_{j^*}(t) P^{T,\max}_{j^*} \Delta t.
  \]
  Both cases imply:
  \[
  z - z' \leq -c^T_{j^*}(t) P^{T,\max}_{j^*} \Delta t.
  \]
  Since
  \[
  c^T_{j^*}(t) = V - \sum_{i=1}^{N} B^-_i(t)\mu_i h^D_{ij^*}(t)
  \geq V - \sum_{i=1}^{N} B^{\max}_i \mu_i h^D_{ij^*}(t),
  \]
  we conclude that
  \[
  z - z' \leq
  \left(\sum_{i=1}^{N} B^{\max}_i \mu_i h^D_{ij^*}(t) - V\right)
  P^{T,\max}_{j^*}\Delta t
  \leq C_2.
  \]
  Combining with $z' \leq z^*$ yields the desired result.
\end{IEEEproof}

Lemma \ref{lemma:gap} establishes the performance gap for each individual time slot. 
Building upon this result, the overall performance gap of the proposed online scheduling algorithm can be derived as follows:
\begin{theorem}
  \label{theorem:performance}
  The long-term average energy consumption of the proposed algorithm satisfies:
  \begin{equation*}
    \lim_{T\to\infty} \frac{1}{T} \sum_{t=0}^{T-1} \sum_{j=1}^{M} \mathbb{E}\left\{ E^T_j(t) + E^C_j(t) \right\} 
    \leq \widebar{E}^* + \frac{C_1 + C_2}{V},
  \end{equation*}
  where $\widebar{E}^*$ is the optimal average energy consumption of the system, and $C_1$ and $C_2$ are constants defined in 
  Lemma~\ref{lemma:dpp_bound} and Lemma~\ref{lemma:gap}, respectively. 
  Moreover, the time-average queue length satisfies:
  \begin{align}
    \lim_{T\to\infty} & \frac{1}{T} \sum_{t=0}^{T-1} \sum_{i=1}^{N} \mathbb{E} \left\{ Q_i(t) \right\} \notag \\
    & \leq \frac{C_1 + C_2 + V \left(P^{T, max} \Delta t + \sum_{i=1}^N \eta \phi_i \lambda_i \right) }{\zeta},
    \label{eq:queue_length}
  \end{align}
  where $\zeta$ denotes the margin to the network stability region, i.e., the queues can be stabilized under arrival 
  rates $\lambda_i + \zeta$ for all $i \in \mathcal{N}$.
\end{theorem}
\begin{IEEEproof}
  The proof directly follows from Theorem~4.8 in~\cite{neely2010stochastic}, 
  with the fact that the per-slot WPT energy consumption is upper bounded by $P^{T,\max} \Delta t$, 
  while the energy consumption due to computation at APs is upper bounded by $\sum_{i=1}^{N} \eta \phi_i \lambda_i$, 
  and $\widebar{E}^{\min}\geq 0$, which yields~\eqref{eq:queue_length}.
  %
\end{IEEEproof}

According to the Little's Law \cite{little1961proof}, 
the average delay of tasks is linearly proportional to the average queue length.
Hence, Theorem~\ref{theorem:performance} establishes an $\mathcal{O}(1/V)$--$\mathcal{O}(V)$ trade-off between system energy efficiency and average delay.

\section{Algorithmic Enhancements for Improved Stability and Latency} \label{section:improvement}

\subsection{Practical Considerations on Queue Magnitudes}
During the experimental evaluation, two practical issues were identified concerning the relative magnitudes of different system metrics. 
First, the battery capacity of commercial IoT WDs is typically several orders of magnitude greater than the amount of energy harvested within 
a single time slot. 
For instance, a 220~mAh coin-cell battery has an energy capacity of approximately 2.4~kJ, 
whereas the average energy harvested per slot is generally below $10^{-4}$~J. 
Consequently, the variation in $B^{-}_i(t)$ becomes negligibly small compared to the overall battery capacity, 
limiting the algorithm’s ability to adaptively adjust control decisions based on the instantaneous residual energy of each WD. 
To mitigate this issue, we introduce a \emph{virtual battery capacity} that is significantly smaller than the physical capacity, 
thereby amplifying the observable energy fluctuations and improving the responsiveness of the algorithm.

Second, when measured in Joules, the magnitude of $B^{-}_i(t)$ is substantially smaller than that of the data queue length $Q_i(t)$. 
As a result, in the process of minimizing queue lengths, the optimization is dominated by $Q_i(t)$, rendering the effect of $B^{-}_i(t)$ marginal. 
To balance the influence of these two queues, we introduce scalar coefficients $\beta_Q$ and $\beta_B$ to rescale $Q_i(t)$ and $B^{-}_i(t)$ respectively, 
ensuring that their magnitudes are comparable.
Accordingly, the quadratic Lyapunov function in~\eqref{eq:lyapunov_function} is redefined as
\begin{equation*}
  L(t) = \frac{1}{2} \sum_{i=1}^{N} \left[ \left( \beta_Q Q_i(t) \right)^2 + \left( \beta_B B^-_i(t) \right)^2 \right].
\end{equation*}
This modification preserves the stability framework while enhancing the sensitivity of the algorithm to variations in both data and energy queues.

\subsection{Improved Latency with Place-Holder Backlogs}
In Lyapunov-based stochastic optimization, queue backlogs can be interpreted as stochastic analogs of Lagrange multipliers in classical static 
convex optimization problems~\cite{eryilmaz2007fair,5291609}. 
These queue variables must be sufficiently large to convey meaningful information to the stochastic optimizer regarding optimal control actions. 
However, in many cases, it is possible to \emph{virtually inflate} the perceived queue length, 
thereby guiding the optimizer toward near-optimal decisions while maintaining smaller actual queue sizes. 
This approach allows the system to achieve comparable performance with significantly reduced queueing delay.

To illustrate this concept, consider the data queue $Q_i(t)$ at WD $i$. 
Fig.~\ref{fig:Q_dynamics} depicts the temporal evolution of a representative $Q_i(t)$. 
It can be observed that the queue length stabilizes after several time slots and oscillates around a steady-state value, 
which corresponds to the optimal Lagrange multiplier~\cite{5291609}, denoted by $q^{*}_{V,i}$. 
Since control decisions depend on the \emph{queue length} rather than the actual number of bits in the queue, 
replacing a portion of the real bits with \emph{virtual} or \emph{placeholder} bits does not alter the control behavior 
but effectively reduces the real queue occupancy and thus the average delay.

Let $Q^{\mathrm{act}}_i(t)$ denote the actual number of bits in the queue of WD $i$, and introduce $q_i$ placeholder bits (i.e., virtual backlog elements). 
The effective queue length used in control decisions is then expressed as
\begin{equation}
  Q_i(t) = Q^{\mathrm{act}}_i(t) + q_i.
  \label{eq:Q_ph_update}
\end{equation}
The key question is how to determine an appropriate value of $q_i$. 
A straightforward choice is to find a constant $q^0_i$ such that there exists a finite time $t_0$ satisfying $Q_i(t) \geq q^0_i$ for all $t \geq t_0$. 
Setting $q_i = q^0_i$ at initialization (i.e., $Q_i(0) = q_i$) ensures that the system’s 
long-term average energy consumption remains unchanged while the average queue backlog is reduced by $q^0_i$.

However, in practice, $q^0_i$ may be relatively small, and the average backlog may still scale as $\mathcal{O}(V)$. 
To achieve a more pronounced improvement, we can relax the requirement that $Q_i(t)$ never falls below $q_i$, 
and instead choose a larger $q_i$ that allows rare deviations below this threshold. 
The work in~\cite{5291609} demonstrates that the probability of large deviations of $Q_i(t)$ from its 
steady-state value $q^{*}_{V,i}$ decays exponentially with the deviation magnitude. 
Consequently, it can be shown that the average system energy consumption remains asymptotically unchanged if the placeholder is set as
\[
q_i = \max \left[ q^{*}_{V,i} - \log^2(V),\, 0 \right].
\]
This adjustment effectively reduces the average backlog by $q_i$, 
resulting in an overall average backlog scaling as $\mathcal{O}(\log^2(V))$. 
Therefore, the original $\mathcal{O}(1/V)$--$\mathcal{O}(V)$ trade-off between energy efficiency and 
delay is improved to an $\mathcal{O}(1/V)$--$\mathcal{O}(\log^2(V))$ trade-off \cite{5291609}.

\begin{algorithm}[t]
  \caption{Improved Algorithm with Place-Holder} 
  \label{alg:place-holder}
    Initialization: $Q^{act}_i(0) = 0, \hat{q}_{V,i}(0) = 0$\;
    \For{$t = 0$ \KwTo $T$}{
      Obtain control decisions by solving problem \eqref{prob:lya}\;
      Update $Q^{act}_i(t)$ with \eqref{eq:Q_update}\;
      Udpate $\hat{q}_{V,i}(t)$ with \eqref{eq:place-holder}\;
      Update $Q_i(t)$ with \eqref{eq:Q_ph_update}\;
    }
\end{algorithm}

A prerequisite for this method is knowledge of the optimal Lagrange multiplier $q^{*}_{V,i}$. 
A common approach is to execute the algorithm without placeholders and estimate $q^{*}_{V,i}$ based on observed queue lengths. 
However, this approach is computationally expensive, 
as it requires rerunning the algorithm over a sufficiently long horizon whenever system parameters or $V$ change. 
Moreover, such reinitialization may be impractical in real-world systems that must operate continuously.

To address this limitation, we employ an \emph{exponential moving average} (EMA) to dynamically estimate $q^{*}_{V,i}$ 
online by tracking the evolution of $Q_i(t)$. 
The estimated value at slot $t$, denoted by $\hat{q}_{V,i}(t)$, is updated as
\begin{equation}
  \hat{q}_{V,i}(t) = (1-\alpha)\hat{q}_{V,i}(t-1) + \alpha Q_i(t),
  \label{eq:place-holder}
\end{equation}
where $\alpha \in (0,1)$ is a learning rate parameter controlling the tracking responsiveness. The corresponding placeholder is then defined as
\begin{equation*}
q_i(t) = \max\left[ \hat{q}_{V,i}(t) - r \log^2(V),\, 0 \right],
\end{equation*}
where $r$ is a tuning parameter to adjust the placeholder length.

The overall procedure is summarized in Algorithm~\ref{alg:place-holder}. 
Initially, both $Q^{\mathrm{act}}_i(t)$ and $\hat{q}_{V,i}(t)$ are set to zero. 
At each time slot, control decisions are obtained using the algorithm described in Section~\ref{section:algorithm design}. 
Subsequently, the actual queue length $Q^{\mathrm{act}}_i(t)$ is updated according to~\eqref{eq:Q_update}, ensuring nonnegativity. 
The estimated multiplier $\hat{q}_{V,i}(t)$ is then updated via~\eqref{eq:place-holder}, 
followed by the computation of the placeholder-adjusted queue length using~\eqref{eq:Q_ph_update}.

\section{Simulation Results} \label{section:simulation}
In this section, we evaluate the performance of the proposed algorithm and compare it with several benchmark algorithms under a range of system settings. 
To facilitate reproducibility and enable further research based on this work, all source code has been made publicly available.


\begin{table}[!t]
\renewcommand{\arraystretch}{1.2}
\caption{Simulation Parameters}
\label{table:parameter}
\centering
\begin{tabularx}{0.82\linewidth}{l | c}
\hline
\textbf{Parameter} & \textbf{Value}\\
\hline
Spectrum bandwidth $B$ & $0.1$~MHz \\
Task arrival rate $A_i(t)$ & $\mathcal{U}[1,2]$~kb \\
Energy harvesting efficiency $\mu_i$ & $0.51$ \\
Energy efficiency of WDs' processors $\kappa_i$ & $10^{-28}$ \\
CPU cycles for one bit of data $\phi_i$ & $10^3$~cycles/bit \\
Noise power at APs $\sigma_j^2$ & $10^{-9}$~W \\
Communication overhead $v_i$ & $1.1$ \\
Energy efficiency of APs' processors $\eta$ & $10^{-9}$~J/cycle \\
Maximum CPU frequency of WDs $f^{max}_i$ & $0.5$~GHz \\
Maximum WPT power $P^{T,max}_j$ & $3$~W \\
Maximum offloading power $P^{O,max}_i$ & $0.1$~W \\
Virtual battery capacity $B^{max}_i$ & $2\times 10^{-3}$~J \\
\hline
\end{tabularx}
\end{table}

\subsection{Simulation Setup}
We consider a WPMEC network consisting of $N=30$ WDs and $M=5$ APs deployed within a $10~\text{m} \times 10~\text{m}$ square area. 
The WDs are randomly distributed, whereas the APs are placed evenly so as to minimize the maximum distance from any point in the region, 
thereby avoiding scenarios in which some WDs are excessively distant from all APs. 

Following \cite{zhu2020computation}, we adopt a simplified Rayleigh fading channel model. 
The uplink channel gain from WD $i$ to AP $j$ is given by $h^U_{ij} = \theta^U d_{ij}^{-2} \lvert \bar{h}_{ij} \rvert^2$,
where $\theta^U = 5 \times 10^{-4}$, $d_{ij}$ denotes the distance between WD $i$ and AP $j$ (in meters), 
and $\bar{h}_{ij}$ is a complex random variable drawn from the standard complex normal distribution $\mathcal{CN}(0,1)$. 
The downlink channel gain is modeled in a similar manner, with $\theta^D = 10^{-3}$.

Unless otherwise specified, each simulation is conducted over $T = 10^4$ time slots, 
with a slot duration of $\Delta t = 10~\text{ms}$. 
The scaling coefficients associated with $Q_i(t)$ and $B_i^{-}(t)$ are set to $\beta_Q = 3 \times 10^{-7}$ and $\beta_B = 10^{10}$, respectively. 
To demonstrate the sustainability of the proposed algorithm, all WDs are assumed to have empty batteries at the beginning of the simulation. 
The learning rate for $\hat{q}_{V,i}(t)$ is chosen as $\alpha = 3 \times 10^{-4}$, and the tuning parameter for the placeholder length is set to $r = 50$. 
The remaining parameter values are summarized in Table~\ref{table:parameter}.

In the following, we evaluate the proposed algorithm with and without placeholders, and compare its performance with the following three benchmark schemes:
\begin{itemize}
    \item \emph{Local Computing Only (LCO):} all computation tasks generated by the WDs are processed locally;
    \item \emph{Fully Offloading (FO):} all computation tasks are entirely offloaded to the APs;
    \item \emph{Myopic:} each WD attempts to process as much data as possible in each time slot using its remaining energy.
\end{itemize}

\subsection{Queue Dynamics}
\begin{figure}[t]
    \centering
    \includegraphics[width=0.4\textwidth]{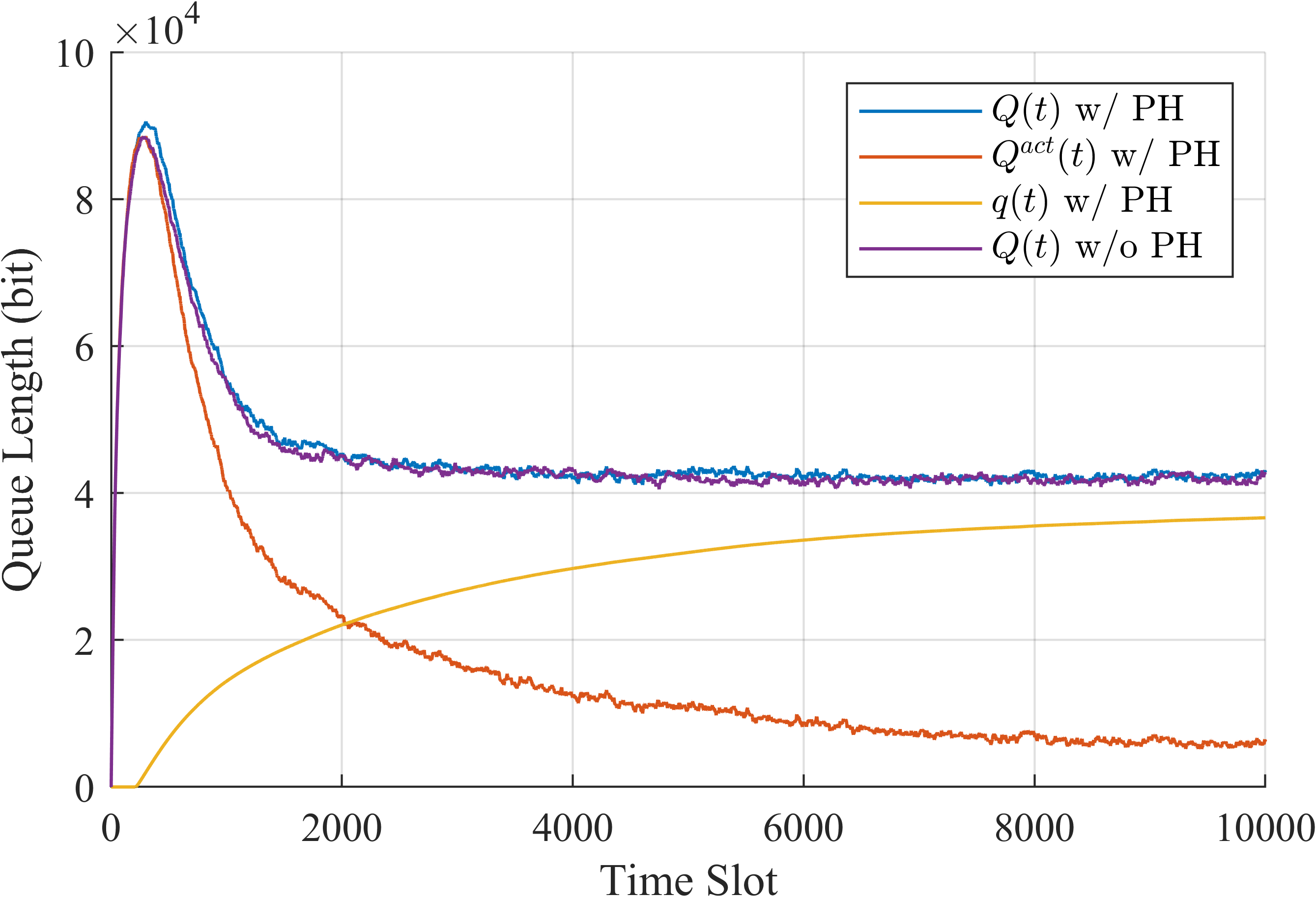}
    \caption{Dynamics of average data queue length.}
    \label{fig:Q_dynamics}
\end{figure}

\begin{figure}[t]
    \centering
    \includegraphics[width=0.4\textwidth]{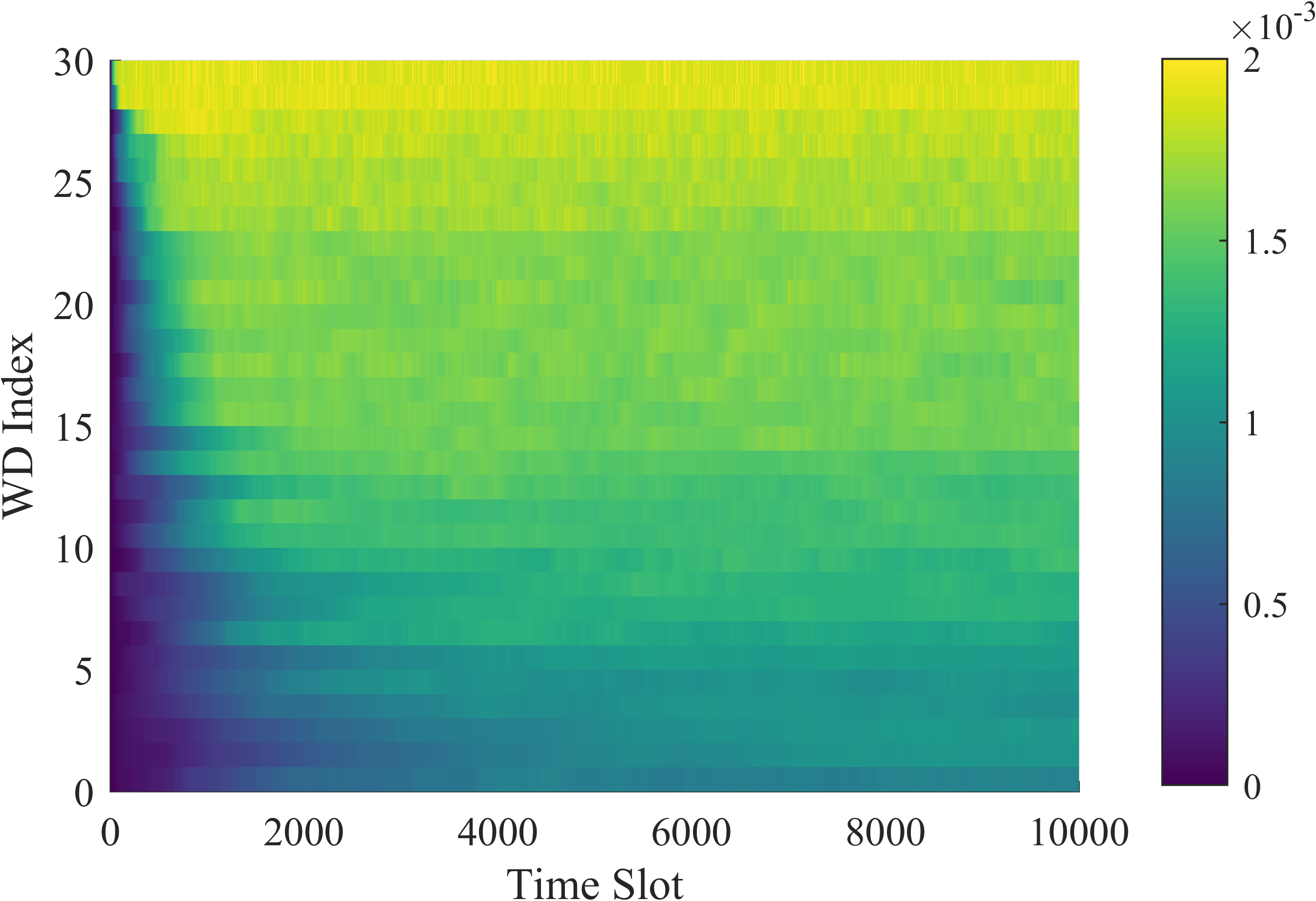}
    \caption{Dynamics of remaining energy at each WD.}
    \label{fig:B_heatmap}
\end{figure}

To illustrate the benefits of introducing placeholders, Fig.~\ref{fig:Q_dynamics} depicts the 
average queue length of all WDs with and without placeholders, denoted by w/ PH and w/o PH, respectively. 
Recall that at the initial stage, the batteries of all WDs are empty, hence 
minimizing $B_i^{-}(t)$ dominates the objective and leads to a rapid growth of the data queues. 
As time progresses, $B_i^{-}(t)$ gradually decreases while $Q_i(t)$ continues to grow, 
causing the reduction of $Q_i(t)$ to become increasingly significant. 
Consequently, $Q_i(t)$ begins to decrease and eventually stabilizes around a steady-state level.

Notably, the queue length $Q(t)$ used for decision making remains unchanged when placeholder bits are introduced. 
As a result, adding placeholders does not affect the per-slot control decisions, hence the corresponding energy consumption remains identical. 
Indeed, we also evaluated the remaining energy of WDs over time and observed no difference between the cases with and without placeholders.
By examining the evolution of the placeholder length $q(t)$, we observe that it remains zero during the initial period. 
This is because $\hat{q}_{V,i}(t)$ is smaller than $r \log^2(V)$ at the beginning, resulting in $q_i(t)=0$. 
Subsequently, the placeholder length $q(t)$ gradually increases and converges toward $Q(t) - r\log^2(V)$. 
As a consequence, the actual queue length $Q^{\mathrm{act}}(t)$ with placeholders is significantly smaller than the actual queue length $Q(t)$ without placeholders,
which leads to a substantial reduction in average latency.
In the following, placeholders are enabled by default in the proposed algorithm.

In addition to data queues, Fig.~\ref{fig:B_heatmap} illustrates the dynamics of the energy queues, 
i.e., the remaining energy levels of the WDs, sorted according to their average energy over time. 
Owing to the rapid attenuation of RF energy with propagation distance, the energy levels among different WDs exhibit substantial heterogeneity. 
Our simulations indicate that the overall energy consumption for WPT is primarily dominated by WDs located far from all APs. 
This observation suggests that, in relatively large deployment areas, 
it is crucial to deploy multiple APs and optimize their locations so as to minimize the maximum distance between any WD and its nearest AP.

\subsection{Impact of $V$}
\begin{figure}[t]
    \centering
    \includegraphics[width=0.4\textwidth]{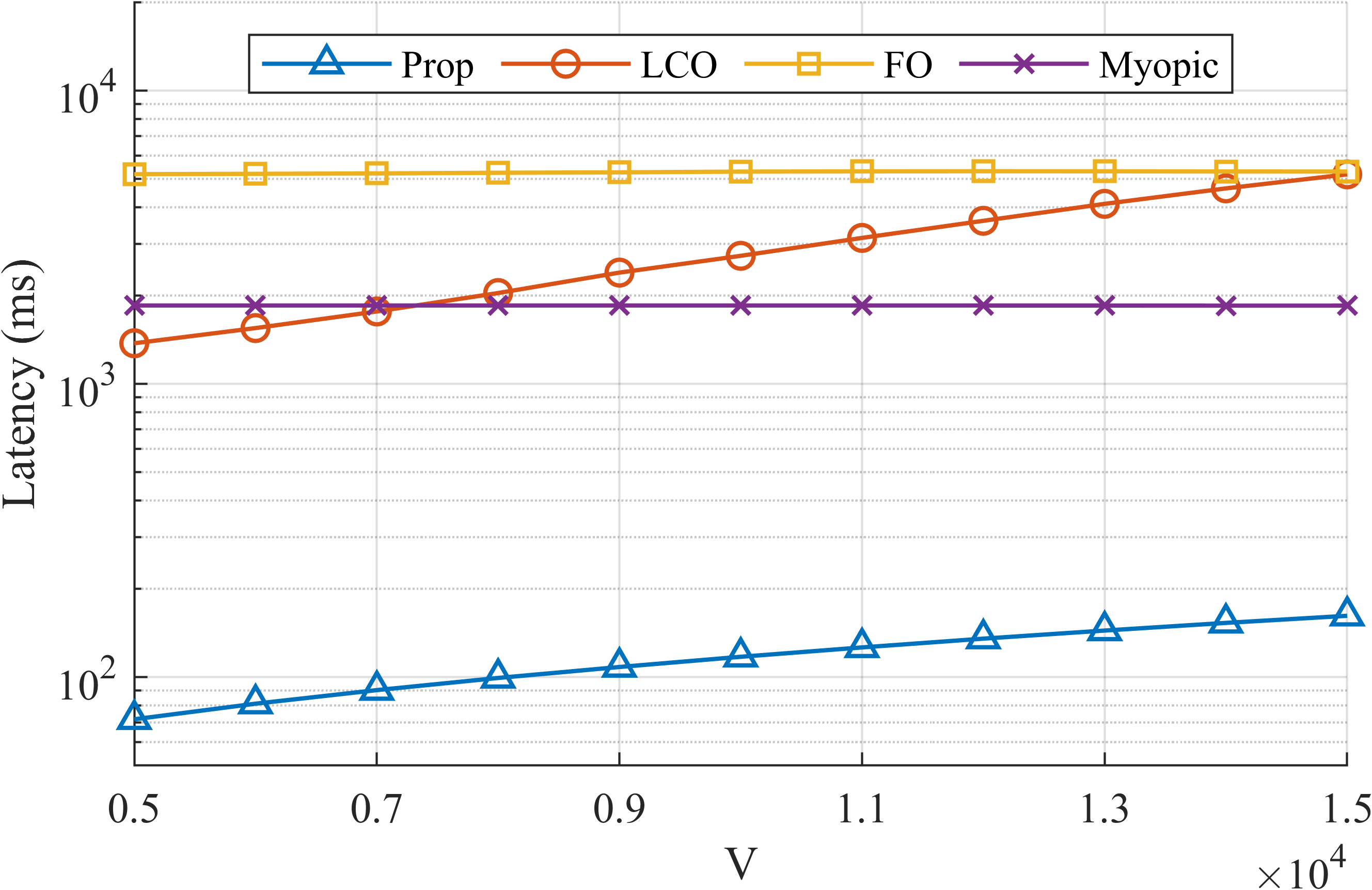}
    \caption{Average latency under different V.}
    \label{fig:V_latency}
\end{figure}

\begin{figure}[t]
    \centering
    \includegraphics[width=0.4\textwidth]{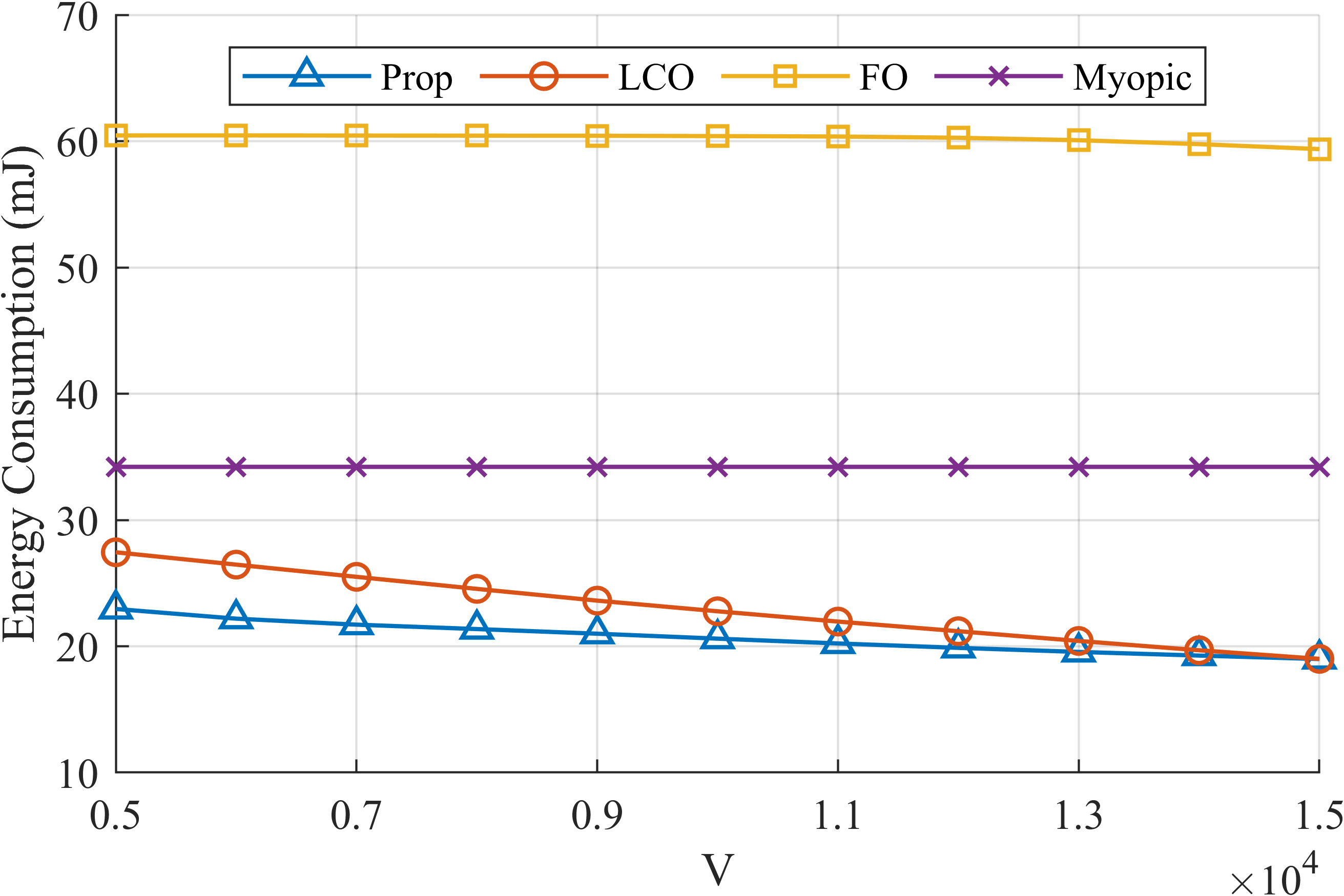}
    \caption{Energy consumption under different V.}
    \label{fig:V_energy}
\end{figure}

The control parameter $V$ plays a critical role in regulating the trade-off between energy consumption and average latency. 
To illustrate its impact, Fig.~\ref{fig:V_latency} and Fig.~\ref{fig:V_energy} present the average latency and energy consumption
of all considered algorithms under different values of $V$. 
Since the original data arrival rate is excessively high for the benchmark schemes, the arrival rate is reduced by $25\%$ in this and following set of experiments. 
To mitigate the effect of randomness, each data point shown in the figures is obtained by averaging the results of ten independent runs with different random seeds.

Among the four algorithms, the FO scheme exhibits the poorest performance in terms of both latency and energy consumption. 
This behavior is primarily attributed to the severe degradation of channel gains experienced by WDs located far from the APs. 
As a result, these WDs harvest substantially less energy while simultaneously expending more energy to offload their data. 
This imbalance not only causes persistent data accumulation at the WDs but also forces the APs to continuously perform WPT to recharge these distant devices. 
Furthermore, the data offloaded to the APs must be processed at the edge servers, which incurs an additional amount of energy consumption, 
thereby further degrading the overall efficiency of the FO scheme.

The Myopic algorithm aims to process as many data bits as possible in each time slot. 
Consequently, the remaining energy at each WD is maintained at a consistently low level, 
which in turn compels the APs to perform WPT in nearly every slot. 
Under such operating conditions, varying the value of $V$ does not affect the algorithm’s decisions in 
either the energy-harvesting or computation-scheduling stages. 
As a result, the performance of the Myopic scheme remains essentially unchanged across different values of $V$. 
It is worth noting that, although the Myopic algorithm exploits both local computing and computation offloading, 
its performance is still significantly inferior to that of the proposed algorithm. 
This observation highlights the advantage of an online optimization framework that accounts for long-term system dynamics over purely myopic decision making.

A clear trade-off between energy consumption and latency can be observed for both the LCO and Prop scheme.
Specifically, as $V$ increases, the average latency increases, while the corresponding energy consumption decreases. 
Nevertheless, across all considered values of $V$, the proposed algorithm consistently achieves an average 
latency that is orders of magnitude lower than those of the benchmark schemes, 
while simultaneously attaining the lowest energy consumption. 
These results demonstrate the benefits of jointly leveraging local computing and computation offloading, 
as well as exploiting coordinated planning across multiple time slots.

\subsection{Impact of Arrival Rate}
\begin{figure}[t]
    \centering
    \includegraphics[width=0.4\textwidth]{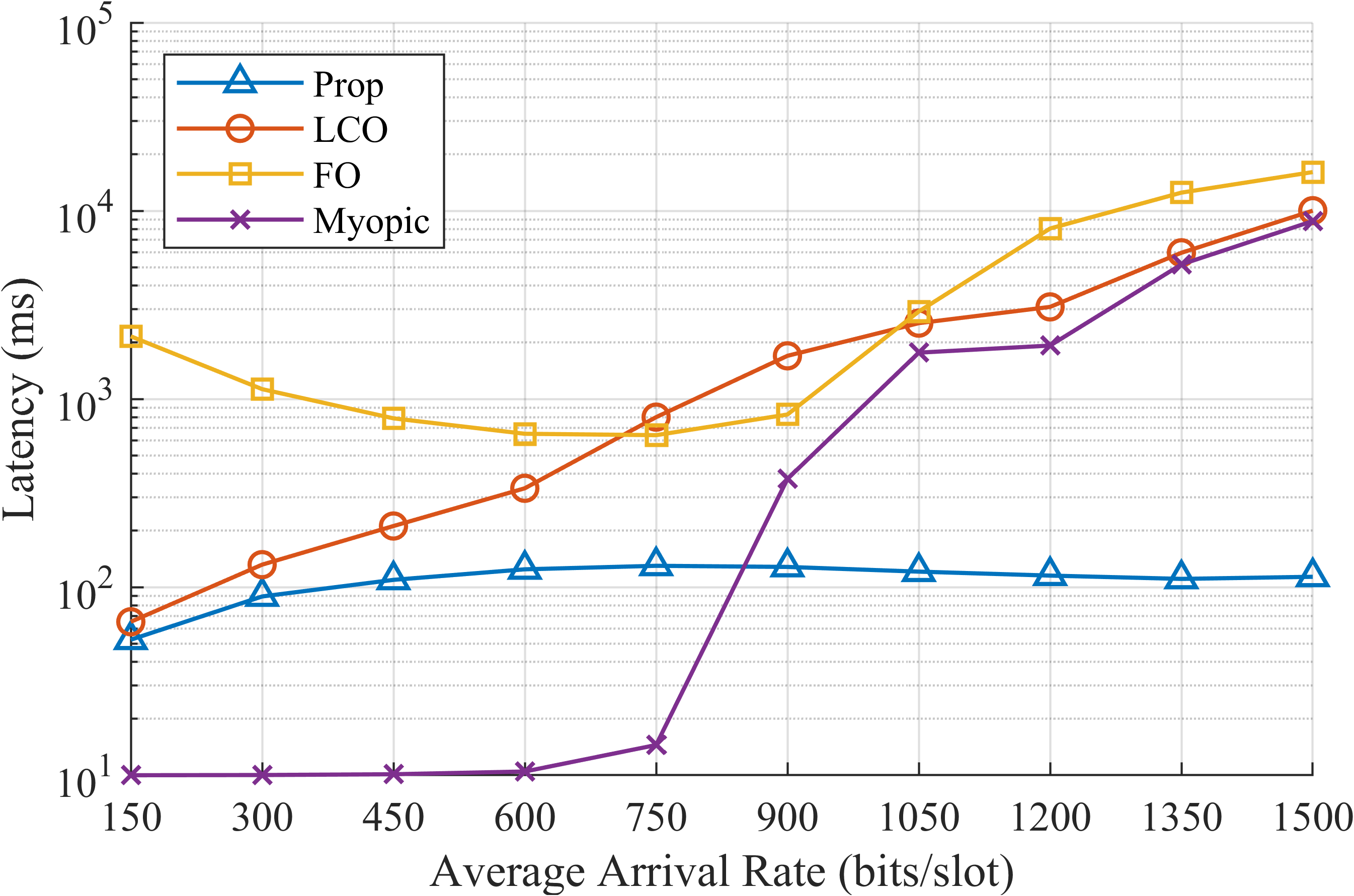}
    \caption{Average latency under different arrival rate.}
    \label{fig:arrival_rate_latency}
\end{figure}

\begin{figure}[t]
    \centering
    \includegraphics[width=0.4\textwidth]{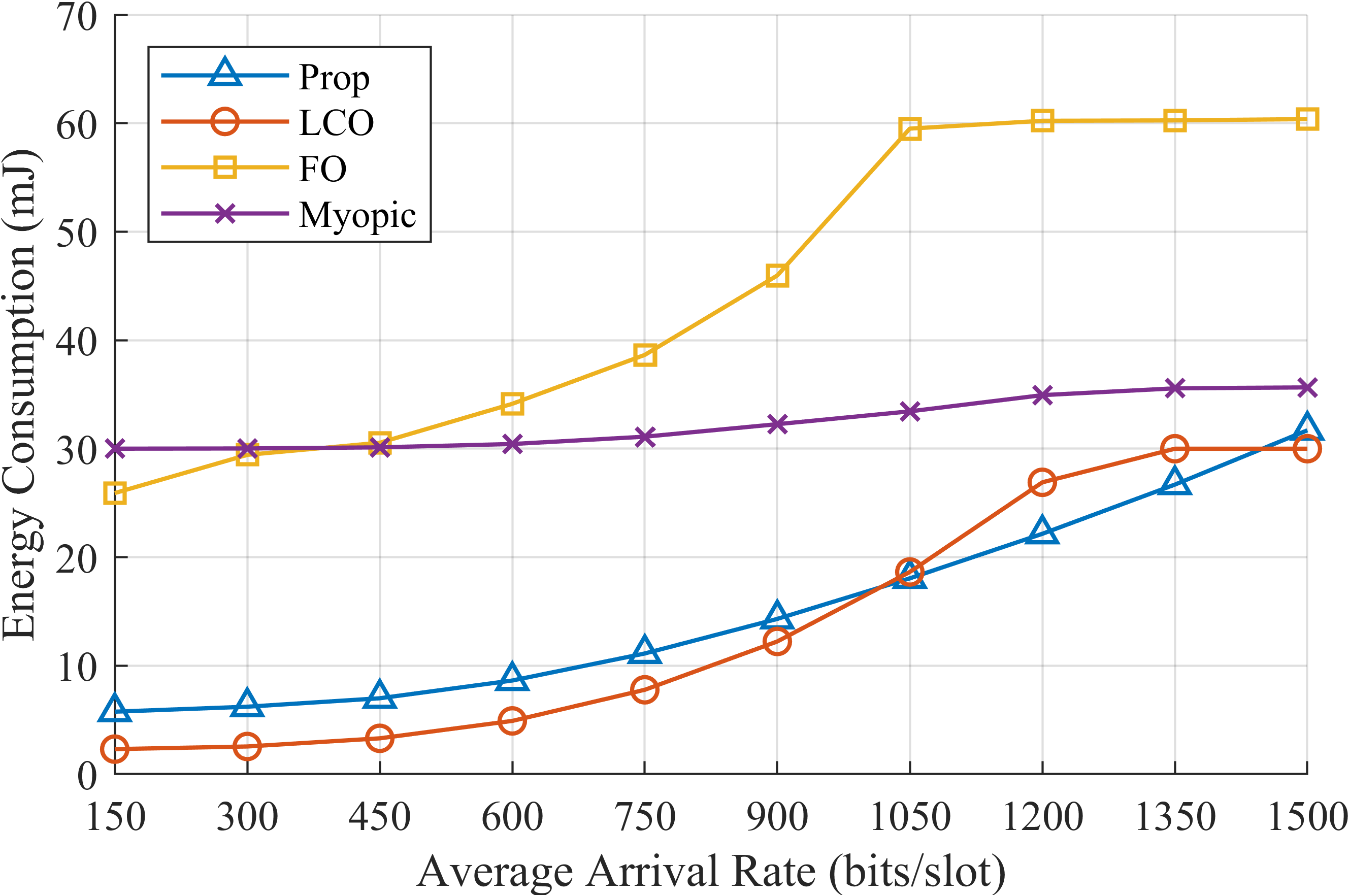}
    \caption{Energy consumption under different arrival rate.}
    \label{fig:arrival_rate_energy}
\end{figure}

To evaluate the performance of all algorithms under different workload conditions, Fig.~\ref{fig:arrival_rate_latency} and 
Fig.~\ref{fig:arrival_rate_energy} illustrate the average latency and energy consumption under various data arrival rate. 
When the workload is very low, the Myopic algorithm is able to process all arriving data within a single time slot, 
resulting in a latency of $10$~ms, which corresponds to the slot duration. 
However, this strategy maintains the remaining energy at the WDs at a consistently low level, thereby forcing the APs to perform WPT in every slot. 
Consequently, the average energy consumption reaches approximately $30$~mJ per slot, corresponding to the WPT power consumption of $3$~W. 
As the arrival rate increases further, the inefficiency of the Myopic scheme becomes apparent: 
the available harvested energy is no longer sufficient to process all incoming data, leading to a rapid increase in latency.

For the FO scheme, recall that the objective function in problem~\eqref{obj:offloading} includes the term $\sum_{i=1}^N Q_i(t) D_i^{O}(t)$. 
As a result, the algorithm increases the offloaded data volume only when $Q_i(t)$ becomes sufficiently large, 
effectively attempting to stabilize the queue length around a certain operating point. 
This explains the observed decrease in latency when the arrival rate increases from a small value,
as a higher arrival rate implies a shorter average waiting time per task under fixed queue length.
However, once the arrival rate exceeds the system capacity, the queue length can no longer be stabilized, and the latency grows rapidly.

When the arrival rate is low, local computing exhibits high energy efficiency. 
Under this condition, the LCO scheme achieves latency and energy consumption comparable to those of the proposed algorithm. 
As the arrival rate increases, however, the energy efficiency of local computing deteriorates rapidly due to the required increase in CPU frequency. 
In particular, for WDs located far from the APs, the harvested energy becomes insufficient to process all arriving data, 
leading to a sharp increase in average latency. 
Nevertheless, since LCO does not offload data to the APs, its maximum energy consumption corresponds to 
performing WPT in every slot, which is bounded by $30$~mJ per slot.

For the Prop algorithm, when the arrival rate is small, the queue lengths are insufficiently large for the placeholder mechanism to take effect. 
As a result, the latency increases slowly with the arrival rate in this regime. 
As the workload continues to grow, the placeholder mechanism maintains the effective queue length around a steady value. 
Consequently, the latency of the proposed algorithm remains relatively stable, while the energy consumption increases with the growing workload.

\subsection{Impact of Network Scale}
\begin{figure}[t]
    \centering
    \includegraphics[width=0.4\textwidth]{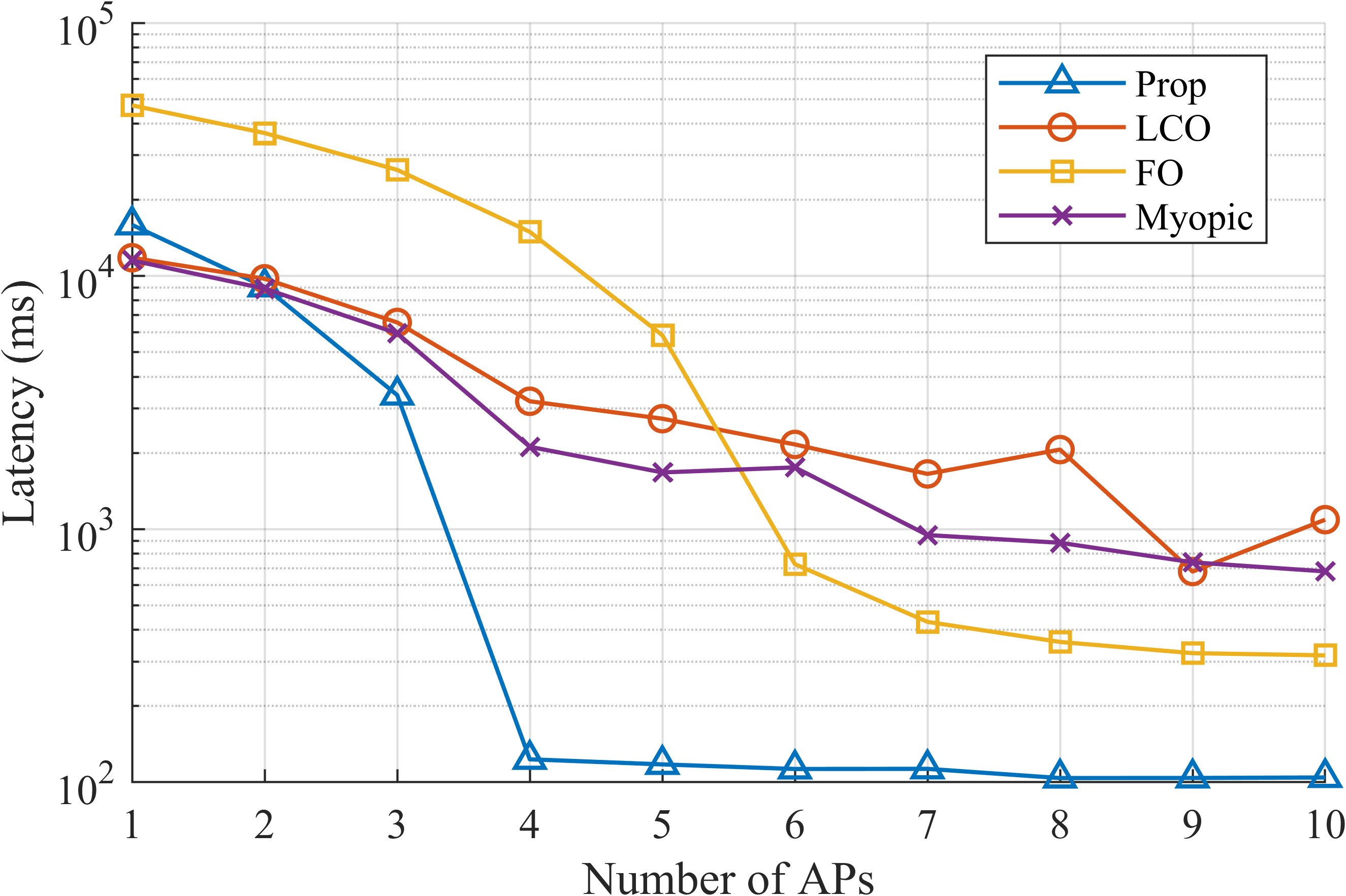}
    \caption{Average latency under different numbers of APs.}
    \label{fig:M_latency}
\end{figure}

\begin{figure}[t]
    \centering
    \includegraphics[width=0.4\textwidth]{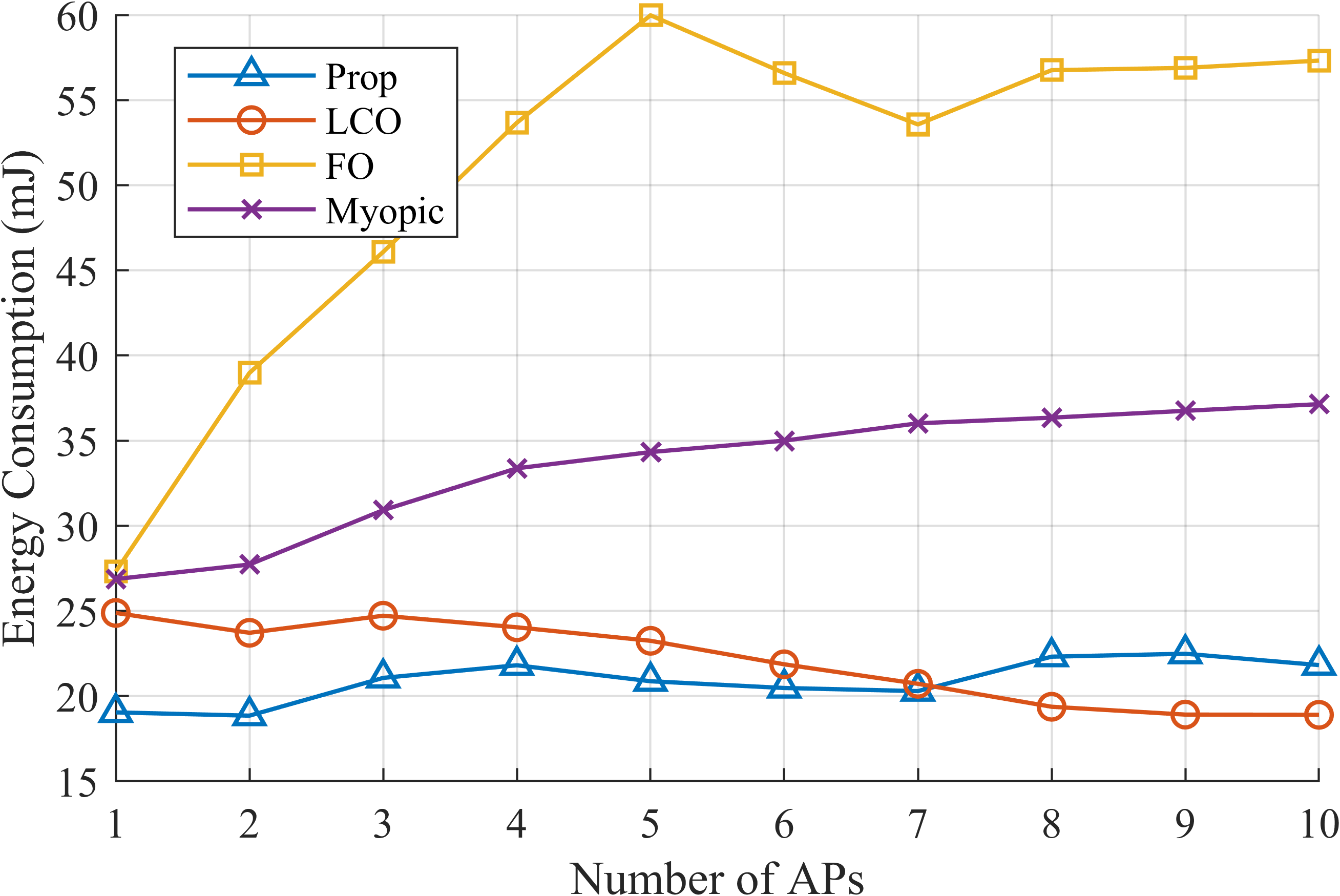}
    \caption{Energy consumption under different numbers of APs.}
    \label{fig:M_energy}
\end{figure}

To examine the performance of all algorithms under different network scales, 
Fig.~\ref{fig:M_latency} and Fig.~\ref{fig:M_energy} depict the average latency and energy consumption under different numbers of APs. 
In all scenarios, the AP locations are chosen to minimize the maximum distance between any point in the considered area and its nearest AP. 
This deterministic placement may introduce certain local irregularities, such as the observed fluctuations in the latency 
of LCO and the energy consumption of FO in the latter part of the curves. 

When only a single AP is deployed, the AP’s time resources are predominantly allocated to WPT, while the WDs mainly rely on local computing to process their data. 
Therefore, the performances of Prop, LCO, and Myopic are similar. 
In contrast, the FO scheme must reserve a portion of the available time for data offloading, which leads to a significantly higher 
average latency compared to the other algorithms. 

As the number of APs increases, performance improvements arise from two main factors. 
First, the efficiency of WPT is enhanced due to reduced propagation distances. 
Second, the efficiency of computation offloading improves as more APs become available. 
Since LCO does not employ computation offloading, its performance gain is solely attributable to the first factor.
Consequently, the performance improvements of FO and Prop are more pronounced than those of LCO. 

\begin{figure}[t]
    \centering
    \includegraphics[width=0.4\textwidth]{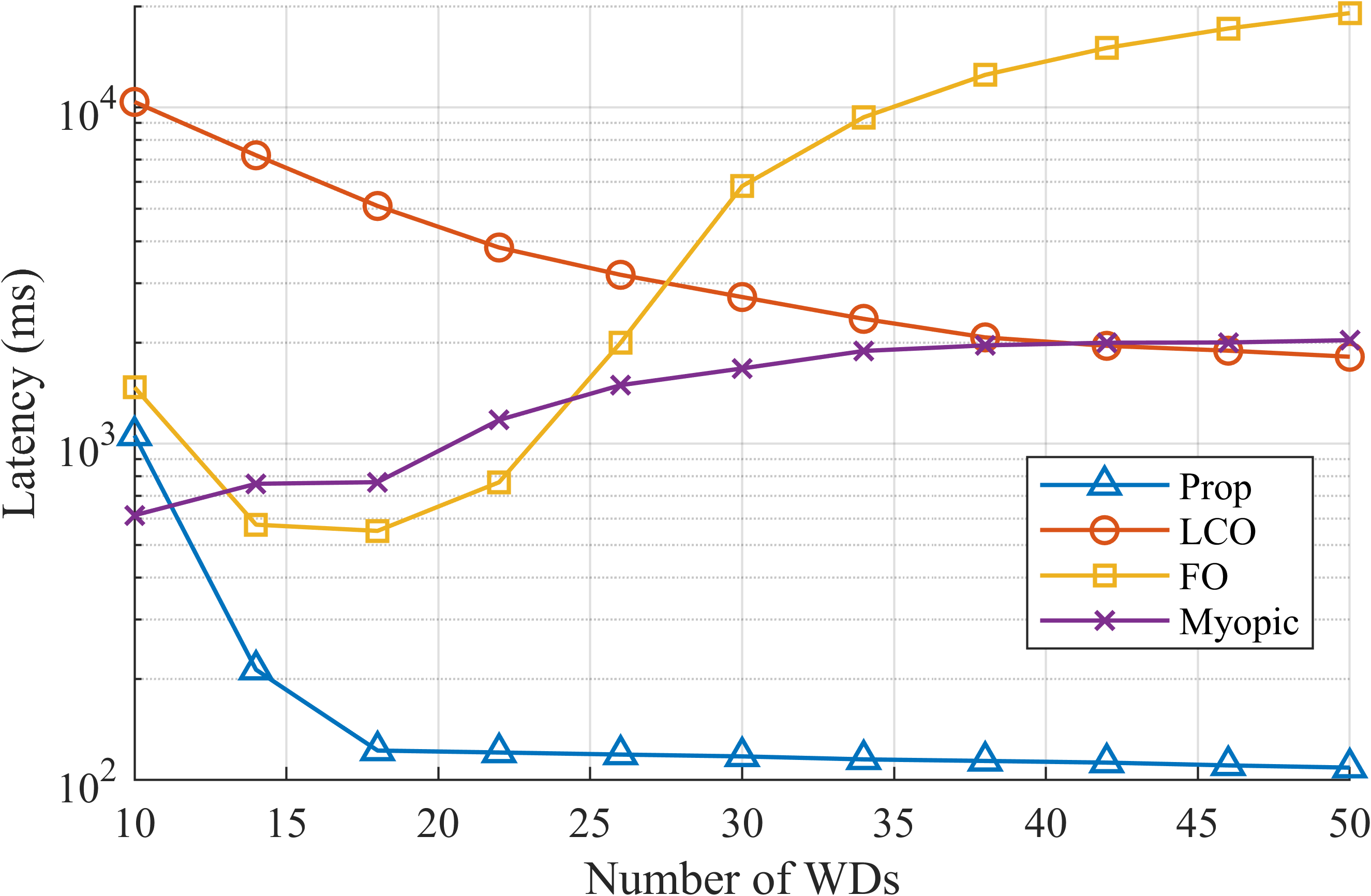}
    \caption{Average latency under different numbers of WDs.}
    \label{fig:N_latency}
\end{figure}

\begin{figure}[t]
    \centering
    \includegraphics[width=0.4\textwidth]{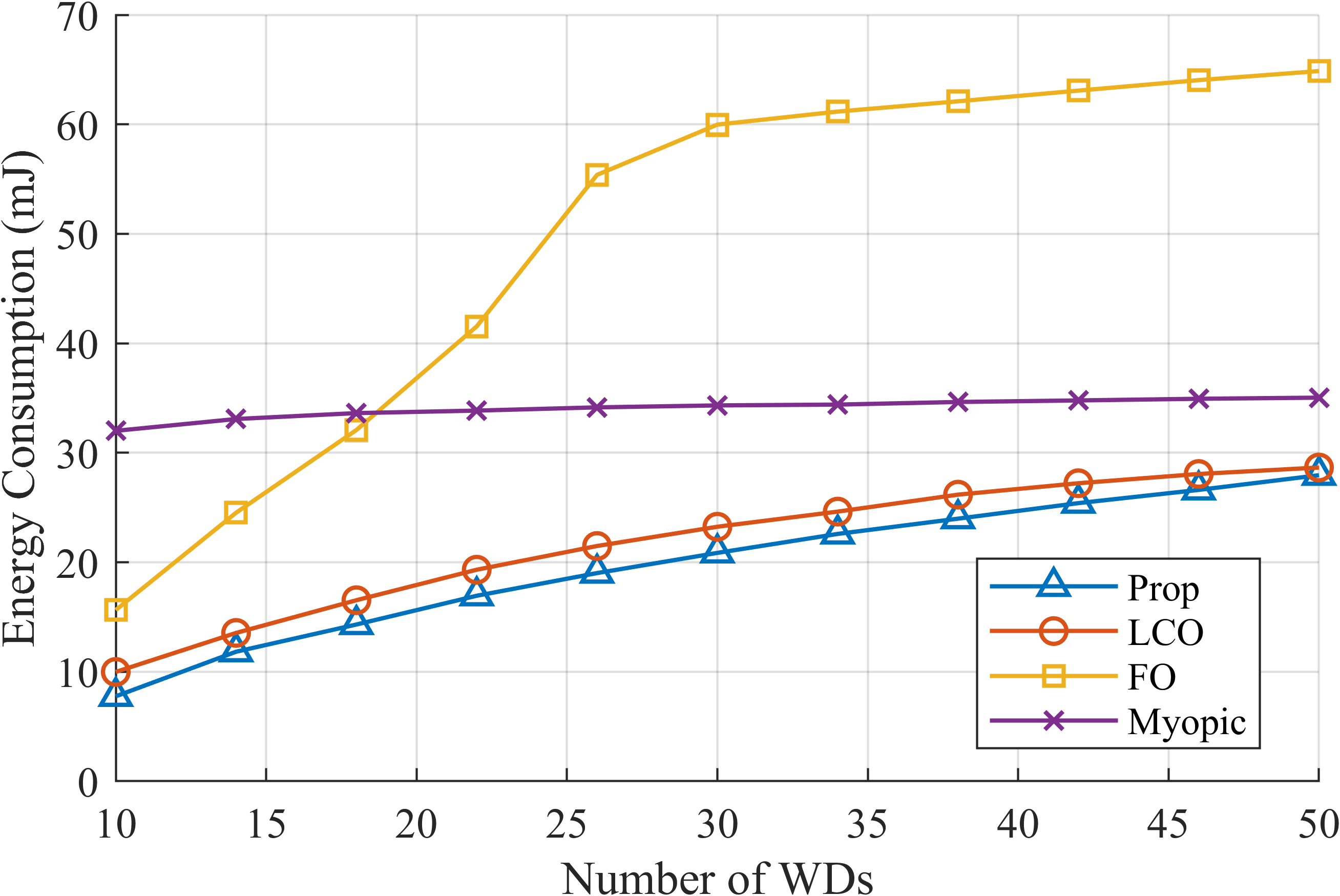}
    \caption{Energy consumption under different numbers of WDs.}
    \label{fig:N_energy}
\end{figure}

Fig.~\ref{fig:N_latency} and Fig.~\ref{fig:N_energy} illustrate the average latency and energy consumption of all algorithms under different numbers of WDs. 
When the number of WDs is small, the latency of Prop, LCO, and FO decreases as the number of WDs increases. 
This phenomenon is due to the same reason explained in the previous subsection:
these algorithms attempt to maintain the total workload in the system around a certain operating point. 
As the number of WDs increases, the aggregate data arrival rate also increases, which leads to reduced average queuing latency. 
However, for the FO scheme, once the number of WDs exceeds 18, the total workload surpasses the system’s processing capacity, 
resulting in a rapid increase in latency. 
In contrast to the above three algorithms, Myopic seeks to minimize the residual workload in each slot rather than regulating the workload around a steady level. 
Consequently, its latency increases monotonically with the number of WDs.

\section{Conclusion} \label{section:conclusion}
This paper studied the energy-efficient scheduling in WPMEC networks with multiple WDs and APs. 
An online optimization framework based on Lyapunov optimization theory was developed to jointly schedule WPT and computation offloading 
while accounting for energy storage and queue dynamics. 
To reduce computational complexity, a relax–then–adjust approach and an efficient per-slot algorithm were proposed. 
Theoretical analysis established explicit performance guarantees for the online algorithm. 
Simulation results demonstrated significant energy savings and latency reduction compared with benchmark schemes.

\bibliographystyle{IEEEtran}
\bibliography{ref}

%
%
%
%
%
%


\end{document}